\documentclass[pra,
superscriptaddress,
longbibliography,
twocolumn
]{revtex4-1}
\usepackage{float}
\usepackage[latin9]{inputenc}
\usepackage{amsmath}
\usepackage{amssymb}
\usepackage{amsthm}
\usepackage{mathtools}
\usepackage{graphicx}
\usepackage{color}
\usepackage[colorlinks=true,citecolor=blue]{hyperref}
\usepackage{times}
\def\identity{\leavevmode\hbox{\small1\kern-3.8pt\Normauyutfrdesawlop;kjhglsize1}}
\usepackage{pgfplots}
\pgfplotsset{
	log x ticks with fixed point/.style={
		xticklabel={
			\pgfkeys{/pgf/fpu=true}
			\pgfmathparse{exp(\tick)}%
			\pgfmathprintnumber[fixed relative, precision=3]{\pgfmathresult}
			\pgfkeys{/pgf/fpu=false}
		}
	},
	log y ticks with fixed point/.style={
		yticklabel={
			\pgfkeys{/pgf/fpu=true}
			\pgfmathparse{exp(\tick)}%
			\pgfmathprintnumber[fixed relative, precision=3]{\pgfmathresult}
			\pgfkeys{/pgf/fpu=false}
		}
	}
}
\newtheorem{theorem}{Theorem}

\newtheorem{definition}{Definition}

\newtheorem{lemma}{Lemma}

\renewcommand{\epsilon}{\varepsilon}
\newcommand{\Norm}[1]{\left\lVert#1\right\rVert}
\newcommand{\norm}[1]{\lVert#1\rVert}
\newcommand{\comm}[1]{\left[#1\right]}
\newcommand{\Abs}[1]{\left|#1\right|}
\newcommand{\abs}[1]{|#1|}
\renewcommand{\O}[1]{O\left(#1\right)}

\usepackage[ruled,linesnumbered]{algorithm2e}
\SetAlgorithmName{Protocol}{protocol}{List of Protocols}
\usepackage{ragged2e}
\usepackage{setspace}

\newcommand\jus[1]{\medskip\par\noindent\textbf{}\justifying} 
\newcommand{\etal}{{et al.}}
\renewcommand{\O}[1]{\mathcal O\left(#1\right)}

\renewcommand{\d}{d}

\newcommand{\mainref}[1]{\ref{#1}}

\newcommand{\Section}[1]{\section{#1}}
\newcommand{\U}[2]{U^{#1}_{#2}}

\newcommand{\Udag}[2]{\left(U^{#1}_{#2}\right)^\dag}
\newcommand{\dtrunc}{\delta_\text{trunc}}
\newcommand{\doverlap}{\delta_\text{overlap}}
\newcommand{\veci}{\vec\imath}
\newcommand{\vecj}{\vec\jmath}
\DeclareMathOperator{\distop}{dist}
\newcommand{\dist}[1]{\distop\left(#1\right)}

\newcommand{\B}{\mathcal{B}}
\newcommand{\lr}{\textnormal{lr}}
\newcommand{\tr}{\textnormal{tr}}
\newcommand{\ov}{\textnormal{ov}}
\usepackage[capitalise,compress]{cleveref}
\crefname{section}{Sec.}{Secs.}
\Crefname{section}{Section}{Sections}
\crefname{appsec}{Appendix}{Appendices}
\crefrangelabelformat{equation}{\textup{(#3#1#4)}--\textup{(#5#2#6)}}

\begin{document}

\title{Locality and digital quantum simulation of power-law interactions}

\author{Minh~C.~Tran}
\affiliation{Joint Center for Quantum Information and Computer Science, NIST/University of Maryland, College Park, Maryland 20742, USA}
\affiliation{Joint Quantum Institute, NIST/University of Maryland, College Park, Maryland 20742, USA}
\affiliation{Kavli Institute for Theoretical Physics, University of California, Santa Barbara, California 93106, USA}
\author{Andrew~Y.~Guo}
\affiliation{Joint Center for Quantum Information and Computer Science, NIST/University of Maryland, College Park, Maryland 20742, USA}
\affiliation{Joint Quantum Institute, NIST/University of Maryland, College Park, Maryland 20742, USA}
\author{Yuan~Su}
\affiliation{Joint Center for Quantum Information and Computer Science, NIST/University of Maryland, College Park, Maryland 20742, USA}
\affiliation{Department of Computer Science, University of Maryland, College Park, Maryland 20742, USA}
\affiliation{Institute for Advanced Computer Studies, University of Maryland, College Park, Maryland 20742, USA}
\author{James~R.~Garrison}
\affiliation{Joint Center for Quantum Information and Computer Science, NIST/University of Maryland, College Park, Maryland 20742, USA}
\affiliation{Joint Quantum Institute, NIST/University of Maryland, College Park, Maryland 20742, USA}
\author{Zachary~Eldredge}
\affiliation{Joint Center for Quantum Information and Computer Science, NIST/University of Maryland, College Park, Maryland 20742, USA}
\affiliation{Joint Quantum Institute, NIST/University of Maryland, College Park, Maryland 20742, USA}
\author{Michael~Foss-Feig}
\affiliation{United States Army Research Laboratory, Adelphi, Maryland 20783, USA}
\affiliation{Joint Center for Quantum Information and Computer Science, NIST/University of Maryland, College Park, Maryland 20742, USA}
\affiliation{Joint Quantum Institute, NIST/University of Maryland, College Park, Maryland 20742, USA}
\author{Andrew~M.~Childs}
\affiliation{Joint Center for Quantum Information and Computer Science, NIST/University of Maryland, College Park, Maryland 20742, USA}
\affiliation{Department of Computer Science, University of Maryland, College Park, Maryland 20742, USA}
\affiliation{Institute for Advanced Computer Studies, University of Maryland, College Park, Maryland 20742, USA}
\author{Alexey~V.~Gorshkov}
\affiliation{Joint Center for Quantum Information and Computer Science, NIST/University of Maryland, College Park, Maryland 20742, USA}
\affiliation{Joint Quantum Institute, NIST/University of Maryland, College Park, Maryland 20742, USA}

\begin{abstract}
	The propagation of information in non-relativistic quantum systems obeys a speed limit known as a Lieb-Robinson bound.  
	We derive a new Lieb-Robinson bound for systems with interactions that decay with distance $r$ as a power law, $1/r^\alpha$.
	The bound implies an effective light cone tighter than all previous bounds.
	Our approach is based on a technique for approximating the time evolution of a system, which was first introduced as part of a quantum simulation algorithm by Haah~{\it et al.}, FOCS'18.
	To bound the error of the approximation, we use a known Lieb-Robinson bound that is weaker than the bound we establish.
	This result brings the analysis full circle, suggesting a deep connection between Lieb-Robinson bounds and digital quantum simulation.
	In addition to the new Lieb-Robinson bound, our analysis also gives an error bound for the Haah~{\it et al.} quantum simulation algorithm when used to simulate power-law decaying interactions.
	In particular, we show that the gate count of the algorithm scales with the system size better than existing algorithms when $\alpha>3D$ (where $D$ is the number of dimensions).
\end{abstract}

\maketitle

\section{Introduction}

Lieb-Robinson bounds limit the rate at which information can propagate in systems that obey the laws of non-relativistic quantum mechanics~\cite{LR,NachtergaeleOS2006,Nachtergaele2006,HK,GongFF,Foss-FeigG,Storch15,NRSS09,SHKM10,SH10}.
These bounds have found a plethora of applications~\cite{Nachtergaele11,BravyiHV06,Cheneau2012,Lashkari2013,Kliesch14,Nachtergaele11,Hamza2012,Barmettler12,Hastings09,Schollwock11,Enss12,Woods2015,Woods2016}, including recent results on entanglement area laws~\cite{Hastings07,Eisert2010,Gong17}, the classical complexity of sampling bosons~\cite{Deshpande17}, and even a quantum algorithm for digital quantum simulation~\cite{Haah}.

Lieb and Robinson's original proof applies only to short-range interactions, i.e.,\ those that act over a finite range or decay at least exponentially in space. 
However, interactions in many physical systems, such as trapped ions~\cite{Britton2012,Kim2011}, Rydberg atoms~\cite{Saffman10}, ultracold atoms and molecules~\cite{Douglas2015,Yan2013}, nitrogen-vacancy centers~\cite{Maze2011}, and superconducting circuits~\cite{Otten16}, can decay with distance $r$ as a power law ($1/r^\alpha$) and, hence, lie outside the scope of the original Lieb-Robinson bound. 
Thus, understanding the fundamental limit on the speed of information propagation  
in these systems holds serious physical implications, including for the applications mentioned above.
Despite many efforts in recent years~\cite{HK,GongFF,Foss-FeigG,Storch15}, a \emph{tight} Lieb-Robinson bound for such long-range interactions remains elusive. 

In this paper, we derive a new Lieb-Robinson bound for systems with power-law decaying interactions in $D$ dimensions.
While our bound is not known to be tight, it has four main benefits compared to the best previous bound for such systems~\cite{Foss-FeigG}:
(i) It is tighter, resulting in the best effective light cone to date [\cref{EQ_LightCone}].
(ii) The bound applies at all times, and not just asymptotically in the large-time limit.
(iii) The framework behind the proof is conceptually simpler, with an easy-to-understand interpretation based on physical intuition.
(iv) Our approach is potentially applicable to studying a wider variety of quantities, including connected correlators~\cite{Bravyi06,Tran17} and higher-order correlators (for instance, the out-of-time-ordered correlator~\cite{Larkin1969,kitaev_soft_2018} and the full measurement statistics of boson sampling \cite{Aaronson2011,Deshpande17}) as we discuss in \cref{Sec_Outlook}.

In contrast to the previous long-range Lieb-Robinson bounds~\cite{HK,GongFF,Foss-FeigG,Storch15}, which all relied on the so-called Hastings-Koma series~\cite{HK}, 
our approach is based on a generalization of the framework Haah \etal~\cite{Haah} (HHKL) introduced as a building block for their quantum simulation algorithm. 
The essence of their framework is a technique for decomposing the time evolution of a system into evolutions of subsystems, with an error bounded by the Lieb-Robinson bound for short-range interactions~\cite{LR}.
We extend the HHKL framework to long-range interactions and to a more general choice of subsystems. 
Remarkably, these modifications enable us to derive a tighter Lieb-Robinson bound for long-range interactions than the one we use in the analysis of the decomposition~\cite{GongFF}.

Additionally, we return to the original motivation of Haah \etal's framework: the digital simulation of lattice-based quantum systems. 
We generalize the HHKL algorithm to simulate systems with power-law decaying interactions. 
The algorithm scales better as a function of system size than previous algorithms when $\alpha>3D$, and the speed-up becomes more dramatic as $\alpha$ is increased.

The structure of the paper is as follows.
In \cref{Sec_Overview}, we state our main results and summarize the proof of the new Lieb-Robinson bound.
In \cref{Sec_Framework}, we lay out the precise mathematical framework for the proof and generalize the technique for decomposing time-evolution unitaries~\cite{Haah} to power-law decaying interactions and to more general choices of subsystems.
After that, we present two applications of the unitary decomposition in \cref{Sec_LR} and \cref{Sec_qu-sim}, which can be read independently of each other.
Specifically, in \cref{Sec_LR}, we use the unitary decomposition to derive the improved Lieb-Robinson bound for long-range interactions.
In \cref{Sec_qu-sim}, we analyze the performance of the HHKL algorithm from Ref.~\cite{Haah} when applied to simulating long-range interacting systems.
We conclude in \cref{Sec_Outlook} with an outlook for the future.
\Section{Summary of results}
\label{Sec_Overview}
In this section, we summarize our main results for the case of a one-dimensional lattice.
Without loss of generality, we assume that the distance between neighboring sites is one.
The unitary decomposition technique in \cref{Sec_Framework} is generalized from a similar result for short-range interactions in Ref.~\cite{Haah}. We use it to approximate the evolution of a long-range interacting system $ABC$ by three sequential evolutions of its subsystems $AB$, $B$, and $BC$ (see Fig.~\ref{FIG_Lem1-demo}). 
We assume that the interaction strength between any two sites in the system is bounded by $1/r^\alpha$, with $r$ being the distance between the sites and $\alpha$ a nonnegative constant.
This restriction on the Hamiltonian norm also sets the time unit for the evolution of the system.

There are two sources of error in the approximation:
one due to the truncation of the Hamiltonian of the system $ABC$ (we ignore the interactions that connect $A$ and $C$), 
and the other due to the Hamiltonians of the subsystems $AB,B$, and $BC$ not commuting with each other. 
For a fixed value of $\alpha$, if the distance $\ell$ between the two regions $A$ and $C$ (see Fig.~\ref{FIG_Lem1-demo}a) is large enough, namely\ $\ell\gg \alpha$, the two error sources have the same scaling with $\ell$.
To estimate the error, for example from the truncation, we sum over interactions connecting sites in $A$ and $C$, and obtain a total error of $\O{1/\ell^{\alpha-2}}$ (in one dimension) for the approximation in the unitary decomposition (as shown in \cref{SM_Subsec_dtrunc}).
\begin{figure}
	\includegraphics[width=0.48\textwidth]{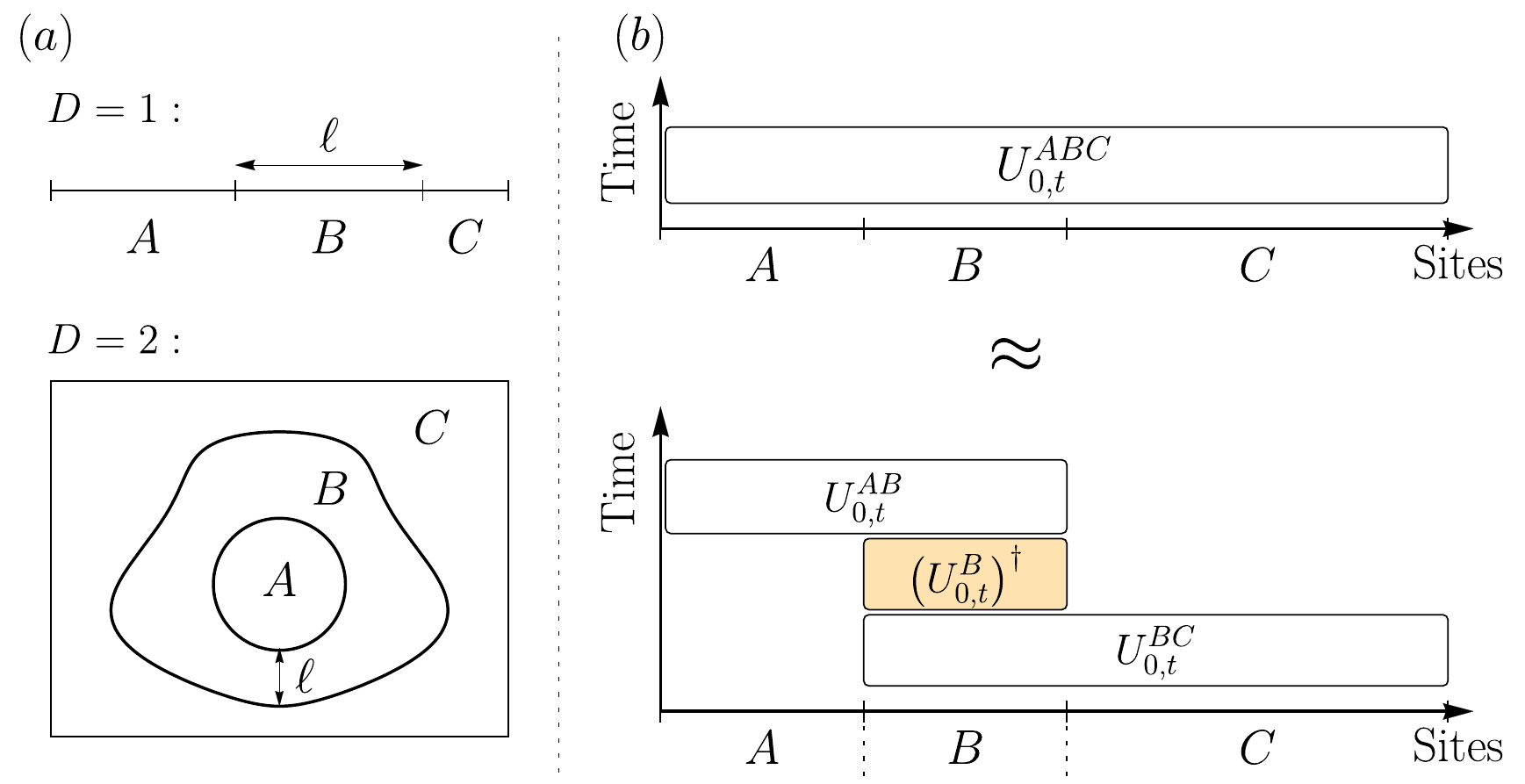}
	\caption{A demonstration of the unitary decomposition in Lemma~\ref{LEM_BREAK_D}. Panel $(a)$: the three disjoint regions $A,B,C$ in $D=1$ and $D=2$ dimensions with $A$ convex and compact. Panel $(b)$: Lemma~\ref{LEM_BREAK_D} allows the evolution of the whole system to be approximated by a series of three evolutions of subsystems. The horizontal axis lists the sites in each of the three sets $A,B,C$ (not necessarily according to their geometrical arrangement, particularly in higher dimensions). Each box is an evolution for time $t$ of a Hamiltonian supported on the sites the box covers. These evolutions can be forward (white fill) or backward (orange fill, with dagger) in time.}
	\label{FIG_Lem1-demo}
\end{figure}

\begin{figure}[t]
	\includegraphics[width=0.5\textwidth]{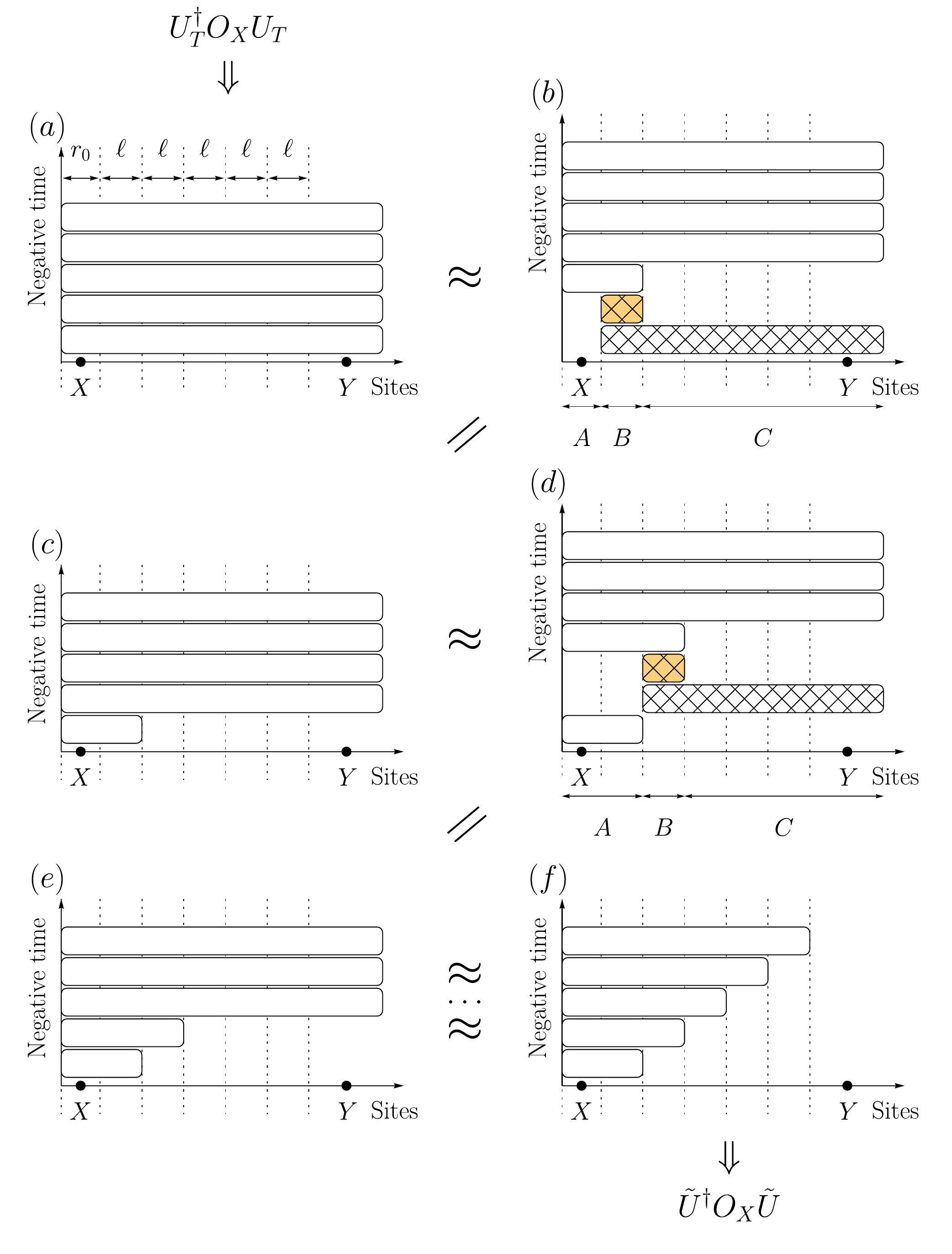}
	\caption{A step-by-step construction of the unitary $\tilde U$ such that $\tilde U^\dag O_X \tilde U \approx U_T^\dag O_X U_T$.  
		Each box represents an evolution of the subsystem covered by the width of the box for a fixed time. 
		The colors of the boxes follow the same convention as in Fig.~\ref{FIG_Lem1-demo}.
		In panel \emph{(a)}, the unitary $U_T$ is written as a product of evolutions of the same system in $M=5$ consecutive time slices.	
		\emph{(b)} The evolution in the last (bottom) time slice is decomposed using the method in Fig.~\ref{FIG_Lem1-demo}, with the choice of subsystems $A,B,C$ such that $X$ is contained in $A$.
		The evolutions of the subsystems $B$ and $BC$ (hatched boxes) therefore commute with $O_X$ and cancel out with their counterparts from $U_T^\dag$, resulting in \emph{(c)}.	
		In panel \emph{(d)}, we repeat the procedure for the second-from-bottom time slice, but note the different choice of $A,B,C$ from panel \emph{(b)}.
		This difference is necessary to ensure that the evolutions of $B$ and $BC$ commute with the evolution(s) from the previously decomposed time slice(s).
		We then commute them through $O_X$ again and remove them from the construction of $\tilde U$ in panel \emph{(e)}.
		Repeatedly applying the unitary decomposition for the other time slices, we obtain the unitary $\tilde U$ in panel \emph{(f)}, which is supported on a smaller region than the original unitary $U_T$. 
		With a proper choice of the size $\ell$ of $B$, we can make sure that $Y$ lies outside this region, and, therefore, $\tilde U$ commutes with $O_Y$.
	}
	\label{FIG_TH-LR-demo-more-step}
\end{figure}

In \cref{Sec_LR}, we use the unitary decomposition to prove a Lieb-Robinson bound for long-range interactions that is stronger than previous bounds, including the one we use in the proof of the unitary decomposition. 
The subject of such a bound is usually the norm of the commutator $\Norm{\comm{O_X(T),O_Y}}$ between an operator $O_X(T)=U_T^\dag O_X U_T$ evolved under a long-range Hamiltonian for time $T$ and another operator $O_Y$ supported on a set $Y$ that is at least a distance $R$ away from the support $X$ of $O_X$.
Here, we briefly explain the essence of the proof using a one-dimensional system with fixed $\alpha$ and large enough $R,T\gg\alpha$ as an example.
The strategy is to use the aforementioned unitary decomposition to construct another unitary $\tilde U$ such that \emph{(i)} $\tilde U^\dag O_X \tilde U$ approximates $U_T^\dag O_X U_T$  and \emph{(ii)} $\tilde U^\dag O_X \tilde U$ commutes with $O_Y$, so the commutator norm $\Norm{\comm{O_X(T),O_Y}}$ will be approximately zero, up to the error of our approximation.
For fixed $\alpha$, we consider $M\propto T$ equal time slices and use the unitary decomposition to extract the relevant parts from the evolution $U_T$ in each time slice. 
Each time we decompose a unitary, we choose the subsystems $A,B,C$ so that only $A$ overlaps with the supports of the unitaries from the previous time slices (see \cref{FIG_TH-LR-demo-more-step}), and therefore the evolutions of $B$ and $BC$ can be commuted through $O_X$ to cancel their counterparts from $U_T^\dag$ (\cref{FIG_TH-LR-demo-more-step}b and \cref{FIG_TH-LR-demo-more-step}d):
\begin{align}
	&\Udag{ABC}{} O_X \U{ABC}{} \nonumber\\
	&\approx \Udag{AB}{}\U{B}{}\Udag{BC}{} O_X \U{BC}{}\Udag{B}{}\U{AB}{} \nonumber\\
	&= \Udag{AB}{}O_X \U{AB}{}.
\end{align}
The remaining evolutions that contribute to the construction of $\tilde U$ are supported entirely on a ball of radius $\sim M\ell$ around $X$, where $\ell$ is the size of $B$ and is chosen to be the same in all time slices.
By choosing $\ell\sim R/M$ and $M\ell<R$ so that $Y$ lies outside this ball, the commutator norm $\Norm{\comm{O_X(T),O_Y}}$ is at most the number of time slices multiplied by $\O{1/\ell^{\alpha-2}}$, which is the decomposition error per time slice.
Therefore, we obtain a Lieb-Robinson bound for long-range interactions in one dimension:
\begin{align}
\Norm{\comm{O_X(T),O_Y}}\leq c_{\lr,\alpha}\frac{T}{\ell^{\alpha-2}} =c_{\lr,\alpha} \frac{T^{\alpha-1}}{R^{\alpha-2}},
\end{align} 
where $c_{\lr,\alpha}$ is a constant that may depend on $\alpha$, but not on $T,R$.
Setting the commutator norm to a small constant yields the causal region inside the effective light cone: $T\gtrsim R^{\frac{\alpha-2}{\alpha-1}}$. 
For comparison, the previous best Lieb-Robinson bound produces a light cone $T\gtrsim R^{\frac{\alpha-2}{\alpha}}$~\cite{Foss-FeigG}. 
Our bound is therefore tighter in the asymptotic limit of large $R$ and large $T$, while its proof is substantially more intuitive than in Ref.~\cite{Foss-FeigG}.
A more careful analysis (\cref{Sec_LR}) shows that our light cone also becomes linear in the limit $\alpha\rightarrow\infty$, where the power-law decaying interactions are effectively short-range.
Moreover, our bound works for arbitrary time $T$, while the bound in Ref.~\cite{Foss-FeigG} applies only in the long-time limit.
We provide a more rigorous treatment as well as a bound for $D$-dimensional systems in \cref{Sec_LR}. 

\Cref{Sec_qu-sim} then then discusses the original motivation for the unitary decomposition\textemdash digital quantum simulation\textemdash in the case of long-range interactions that decay as a power law. 
For $\alpha>2D$, our analysis shows that the HHKL algorithm~\cite{Haah} requires only  $\O{{Tn(Tn/\epsilon)^{\frac{2D}{\alpha-D}}}\log\frac{Tn}{\epsilon}}$ two-qubit gates to simulate the evolution of a system of $n$ sites arranged in a $D$-dimensional lattice for time $T$ with an error at most $\epsilon$.
For large $\alpha$, the gate count of the algorithm scales with $n$ significantly better than other algorithms. 

\Section{Framework}%
\label{Sec_Framework}
In this section, we present the technique for approximating the time evolution of a system by evolutions of subsystems. We later use this technique to derive a stronger Lieb-Robinson bound (\cref{Sec_LR}) and an improved quantum simulation algorithm (\cref{Sec_qu-sim}) for systems with long-range interactions.

We consider $n$ sites arranged in a $D$-dimensional lattice $\Lambda\subset \mathbb N^D$ of size $L = \O{n^{1/D}}$ and $D\geq 1$. 
Recall that, without loss of generality, we assume the spacing between neighboring lattice sites is one. 
This assumption sets the unit for distances between sites in the lattice.
We shall embed the lattice $\Lambda$ into the real space $\mathbb R^D$. 
The intersection $X\cap \Lambda$ therefore contains every lattice site in a subset $X\subset \mathbb R^D$.
The system evolves under a (possibly) time-dependent Hamiltonian
$
H_\Lambda(t)=\sum_{\veci,\vecj}h_{\veci,\vecj}(t),
$
with $h_{\veci,\vecj}{(t)}$ being the interaction between two sites $\veci,\vecj \in \Lambda$.
Without ambiguity, we may suppress the time-dependence in the Hamiltonians.
We say a system has power-law decaying interactions if $\Norm{h_{\veci,\vecj}{}}\leq\frac{1}{\Norm{\veci-\vecj}^\alpha}$, where $\Norm{\cdot}$ denotes both the matrix and the vector 2-norms, for some nonnegative constant $\alpha$ and for all $\veci\neq\vecj$. [Note that $h_{\veci,\veci}$ may have arbitrarily large norm.]
For readability, we denote by $H_X=\sum_{\veci,\vecj\in X} h_{\veci,\vecj}$ the terms of $H_\Lambda$ that are supported entirely on a subset $X\cap \Lambda$, and by $U^X_{t_1,t_2}\equiv \mathcal T \exp\left(-i\int_{t_1}^{t_2} H_X dt\right)$ the evolution unitary under $H_X$ from time $t_1$ to $t_2$, where $\mathcal T$ is the time-ordering operator. 
We also denote by $\dist{X,Y}$ the minimum distance between any two sites in $X$ and $Y$, by $X^c=\mathbb R^D\setminus X$ the complement of $X$ in real space, by $\partial X$ the boundary of a compact subset $X$, by $\Phi(X)$ the area of $\partial X$, and by $XY$ the union $X\cup Y$.
In the following, we keep track of how errors scale with time, distance, and $\alpha$, while treating the dimension $D$ as a constant.

We now describe how to approximate the evolution of the system to arbitrary precision by a series of evolutions of subsystems using a technique we generalize from Ref.~\cite{Haah}.
\begin{lemma}
	\label{LEM_BREAK_D}
	Let $A,B,C\subset \mathbb R^D$ be three distinct regions with non-empty interiors such that $A\cup B\cup C = \mathbb R^D$. Let $A$ be both compact (closed and bounded) and convex. 
	We have
	\begin{align}
	&\Norm{\U{ABC}{0,t}-\U{AB}{0,t}\left(\U{B}{0,t}\right)^\dag\U{BC}{0,t}}\leq c_0 (e^{vt}-1) \Phi(A) \xi_\alpha(\ell),	\nonumber
	\end{align}  
	with 
	\begin{align}
	\xi_{\alpha}(\ell) = 
	\left(\frac{16}{1-\gamma}\right)^\alpha\frac{1}{\ell^{\alpha-D-1}}	
	+e^{-\gamma \ell},\label{EQ_xi}
	\end{align}
	for all $\alpha>D+1$. Here, $v,c_0\in\mathbb R^+$ are positive constants, $\gamma$ is a constant that can be chosen arbitrarily in the range $(0,1)$, and $\ell = \dist{A,C}$ is the distance between sets $A$ and $C$.
\end{lemma}	
We emphasize that this lemma applies to arbitrary sets $A$ that are both convex and compact.
The sets we focus on include $D$-balls and hyperrectangles in $\mathbb R^D$. The former geometry is relevant in the proof of our new Lieb-Robinson bound, the latter in the analysis of the HHKL algorithm for long-range interactions. 

Lemma~\ref{LEM_BREAK_D} allows us to approximate the evolution of a long-range interacting system $ABC$ by that of subsystems $AB,B,BC$ (Fig.~\ref{FIG_Lem1-demo}).
The features of the function $\xi_{\alpha}(\ell)$ are better understood by considering two limiting cases of physical interest. 
First, when $\alpha$ is finite and $\ell$ (the distance between $A$ and $C$) is large compared to $\alpha$, the function $\xi_{\alpha}(\ell)$ behaves like
\begin{align}
	\O{\frac{1}{\ell^{\alpha-D-1}}},
\end{align}
which decays only polynomially with $\ell$.
In the second limit, as $\alpha\to\infty$ for a large but finite $\ell$, we recover from $\xi_\alpha(\ell)$ the exponentially decaying error bound $e^{-\gamma \ell}$\textemdash a trademark of finite-range interactions~\cite{LR,Haah}.

The proof of Lemma~\ref{LEM_BREAK_D}, while more general, bears close resemblance to the corresponding analysis for short-range interactions in Ref.~\cite{Haah}.
However, there are two key differences. 
First, in order to make the approximation in Lemma~\ref{LEM_BREAK_D}, some interactions between sites separated by a distance greater than $\ell$ are truncated from the Hamiltonian. 
While such terms vanish in a system with short-range interactions, here they contribute $\O{{\Phi(A)}/{\ell^{\alpha-D-1}}}$ to the error of the approximation.
In addition, instead of the original Lieb-Robinson bound~\cite{LR} which applies only to systems with short-range interactions, we use Gong \etal's generalization of the bound for long-range interactions~\cite{GongFF}.
The result is an approximation error that decays with $\ell$ polynomially as $\O{{\Phi(A)}/{\ell^{\alpha-D-1}}}$, \emph{in addition} to the exponentially decaying error that exists already for short-range interactions.
Nevertheless, the error can always be made arbitrarily small by choosing $\ell$ to be large enough. 

In \cref{SM_Sec_Lem1_Proof} below, we present the proof of Lemma~\ref{LEM_BREAK_D}.
After that, we demonstrate the significance of Lemma~\ref{LEM_BREAK_D} with two applications: a stronger Lieb-Robinson bound for long-range interacting systems (\cref{Sec_LR}) and an improved error bound for simulating these systems (\cref{Sec_qu-sim}).
Both sections are self-contained, and readers may elect to focus on either of them.
\subsection{Error bound on the unitary decomposition}
\label{SM_Sec_Lem1_Proof}

Here, we will outline the proof of Lemma~\ref{LEM_BREAK_D}.
Similar to  Ref.~\cite{Haah}, we begin with an identity:
\begin{align}
\U{ABC}{0,t}=\U{AB}{0,t}\U{C}{0,t}\underbrace{\Udag{C}{0,t}\Udag{AB}{0,t}\U{ABC}{0,t}}_{=W_t}.
\end{align}
Our aim is to approximate $W_t$ by $\Udag{C}{0,t}\Udag{B}{0,t}\U{BC}{0,t}$, from which Lemma~\mainref{LEM_BREAK_D} will follow.
For that, we look at the generator of $W_t$~\cite{Haah}, i.e., a Hamiltonian $\mathcal{G}_t$ such that 
\begin{align}
\frac{dW_t}{dt} = -i \mathcal{G}_t W_t,\label{EQ_G_def}
\end{align}
for all time. 
Exact differentiation of $W_t$ yields~\cite{Osborne06,Michalakis12}
\begin{align}
\mathcal{G}_t &=\Udag{C}{0,t}\Udag{AB}{0,t}
(\underbrace{H_{ABC}-H_{AB}-H_C}_{=H_{A:C}+H_{B:C}}) 
\U{AB}{0,t}\U{C}{0,t}\label{EQ_LEM_BREAK1}\\
&=\Udag{C}{0,t}\Udag{AB}{0,t}H_{B:C}\U{AB}{0,t}\U{C}{0,t}+\dtrunc\label{EQ_LEM_BREAK2}\\
&=\Udag{C}{0,t}\Udag{B}{0,t}H_{B:C}\U{B}{0,t}\U{C}{0,t}+ \doverlap+\dtrunc,\label{EQ_LEM_BREAK3}
\end{align}
where $H_{X:Y}=\sum_{i\in X,j\in Y}h_{ij}(t)$ denotes the sum of terms supported \emph{across} disjoint sets $X$ and $Y$, and $\dtrunc,\doverlap$ are error terms we now define and evaluate.
Note that the first term in Eq.~\eqref{EQ_LEM_BREAK3} is the generator of $\Udag{C}{0,t}\Udag{B}{0,t}\U{BC}{0,t}$\textemdash the unitary with which we aim to approximate $W_t$.

In contrast to the approximation for short-range interacting systems in Ref.~\cite{Haah}, there are two sources of error in Eq.~\eqref{EQ_LEM_BREAK3}. 
The first error term $\dtrunc$ arises after we discard $H_{A:C}$ from Eq.~\eqref{EQ_LEM_BREAK1}.
For the short-range interactions in Ref.~\cite{Haah}, this error vanishes when the distance $\ell$ between $A$ and $C$ is larger than the interaction range.
However, in our case, there is a nontrivial truncation error associated with ignoring long-range interactions between $A$ and $C$:
\begin{align}
\Norm{\dtrunc }= \Norm{H_{A:C}}&=c_{\tr}2^\alpha\frac{\Phi(A)}{\ell^{\alpha-D-1}}\label{EQ_sumAC}
\end{align}
for $\alpha>D+1$, where $c_{\tr}$ is a constant [\cref{EQ_ctr_def}], $\ell=\dist{A,C}$ is the distance between $A$ and $C$. 
The factor of $1/l^{\alpha}$ in the bound comes from the requirement that the two-body interactions decay as a power law $1/r^\alpha$,
while the term $\ell^{D}$ is due to the sum over all sites in the $D$-dimensional set $C$. Another factor of $\ell\Phi(A)$ arises after summing over the volume of $A$, which we assume to be a compact and convex set.
The detailed evaluation of the norm is presented in Appendix~\ref{SM_Subsec_dtrunc}.

The other error, which we define to be $\doverlap$, is the result of the approximation used between \cref{EQ_LEM_BREAK2,EQ_LEM_BREAK3}.
In the former equation, the operator evolves under $H_{AB}+H_C$, whereas in the latter, it evolves under the reduced Hamiltonian $H_{B}+H_C$, thus incurring the error:
\begin{align}
\Norm{\doverlap} &=  \Norm{\Udag{AB}{0,t}H_{B:C}\U{AB}{0,t}-\Udag{B}{0,t}H_{B:C}\U{B}{0,t}}.\label{EQ_doverlap}
\end{align}
To understand why $\norm{\doverlap}$ is small, recall that $H_{B:C}$ is the sum of terms $h_{\vec b,\vec c}$ that are supported on two sites $\vec b\in B$ and $\vec c\in C$. 
Since the strengths of such terms decay as $1/r^\alpha$ (with $r$ the distance between the sites $\vec b$ and $\vec c$), the main contribution to $H_{B:C}$\textemdash and thus to $\doverlap$\textemdash comes from the terms where $\vec b$ and $\vec c$ are spatially close to each other.
But since the sets $A,C$ are separated by a large distance $\ell$, if the site $\vec b$ is close to $C$, then it must be far from $A$. Thus, the evolution of $h_{\vec b,\vec c}$ for a short time under $H_{AB}$ can be well-approximated by evolution under $H_{B}$ alone.
In \Cref{SM_Subsec_doverlap}, we make this intuition rigorous using Gong \etal~\cite{GongFF}'s generalization of the Lieb-Robinson bound to systems with long-range interactions.

In the end, we obtain the following bound on $\doverlap$:
\begin{align}
\Norm{\doverlap}\leq c_{\ov}(e^{vt}-1) \Phi(A)\Bigg[
\frac{\left(\frac{16}{1-\gamma}\right)^\alpha}{\ell^{\alpha-D-1}}
+\frac{1}{e^{\gamma \ell}}\Bigg],
\end{align}
where $c_{\ov}$ is a constant [\cref{EQ_cov_def}] and $\gamma\in(0,1)$ is a free parameter.
The bound has contributions from two competing terms: one that decays polynomially with $\ell$ and another that decays exponentially.
The polynomially decaying term is dominant for fixed $\alpha$ and large $\ell$, whereas the exponentially decaying term prevails as $\alpha \to \infty$ for fixed $\ell$. 
The errors $\dtrunc$ and $\doverlap$ in approximating the generator $\mathcal G_W$ combine to give an overall error in approximating $W_t$ with $\Udag{C}{0,t}\Udag{B}{0,t}\U{BC}{0,t}$ (see \Cref{APP_ERR_PROP}). From this, we obtain the error bound in \Cref{LEM_BREAK_D}, with $c_0 = \max\{c_{\tr},c_{\ov}\}/v$.

Before discussing applications of \Cref{LEM_BREAK_D}, we pause here to note that the Lieb-Robinson bound in Gong \etal~\cite{GongFF} used in the above analysis is not the tightest-known bound for long-range interactions~\cite{Foss-FeigG}. 
Our use of this bound, however, does not lead to a suboptimal error bound in \Cref{LEM_BREAK_D}. For finite $\alpha$, the error bound is dominated by the polynomially decaying term $1/\ell^{\alpha-D-1}$, which arises from the truncation error $\dtrunc$ rather than $\doverlap$. Therefore, this error term would not benefit from a tighter Lieb-Robinson bound.
In the limit $\alpha\to \infty$, on the other hand, we shall see later that the lemma already reproduces the short-range Lieb-Robinson bound, which is optimal up to a constant factor. 
Thus, we expect that using stronger Lieb-Robinson bounds would produce no significant improvement for the error bound in \Cref{LEM_BREAK_D}.
\Section{A stronger Lieb-Robinson bound}
\label{Sec_LR}

In this section, we will use Lemma~\ref{LEM_BREAK_D} to derive a stronger Lieb-Robinson bound for long-range interactions.
The first generalization of the Lieb-Robinson bound to power-law decaying interactions was given by Hastings and Koma~\cite{HK}.
However, their bound diverges in the limit $\alpha \to\infty$, where the power-law decaying interactions are effectively short-range.
Later, Gong \etal~\cite{GongFF} derived a different bound that, in this limit, does indeed converge to the Lieb-Robinson bound for short-range interactions.
While we used this bound in \cref{Sec_Framework} to prove Lemma~\ref{LEM_BREAK_D}, we will also show that by \emph{using} this lemma, we can in turn derive a Lieb-Robinson bound for long-range interactions that is stronger than the one in Gong \etal\ \
In fact, our bound produces a tighter effective light cone than even the strongest Lieb-Robinson bound for long-range interactions known previously~\cite{Foss-FeigG}.

Recall that the subject of a Lieb-Robinson bound is the commutator norm
\begin{align}
\mathcal{C}(T,R)\equiv  \Norm{\comm{\Udag{\Lambda}{0,T}O_X\U{\Lambda}{0,T},O_Y}},\label{EQ_CTR}
\end{align}
where $O_X,O_Y$ are two operators supported respectively on two sets $X,Y$ geometrically separated by a distance $R$, and $\U{\Lambda}{0,T}$ is the time-evolution unitary of the full lattice $\Lambda$ under a power-law decaying Hamiltonian, as defined above. 

To compare different bounds, we analyze their effective light cones, which, up to constant prefactors, predict the minimum time it takes for the correlator $\mathcal C(T,R)$ to reach a certain value.
For example, the original Lieb-Robinson bound~\cite{LR} produces a linear light cone 
$T\gtrsim R$ for short-range interactions.
For long-range interactions, Hastings and Koma~\cite{HK} first showed that $\mathcal C(T,R)\leq c {e^{v T}}/{R^\alpha}$ for some ($\alpha$-dependent) constants $c,v$.
By setting $\mathcal C(T,R)$ equal to a constant, the bound gives an effective light cone $T\gtrsim\log R$
in the limit of large $T$ and $R$.
Gong \etal~\cite{GongFF} later achieved a tighter light cone that is linear for short distances and becomes logarithmic only for large $R$. 
Shortly after, Foss-Feig \etal~\cite{Foss-FeigG} derived a bound with a polynomial light cone:
\begin{align}
T \gtrsim R^{\frac{\alpha-2D}{\alpha-D+1}}.\label{EQ_MikeLR}
\end{align}
\Cref{EQ_MikeLR} was the tightest light cone known previously.

In the remainder of this section, we use Lemma~\ref{LEM_BREAK_D} to derive a Lieb-Robinson bound for long-range interactions that produces an effective light cone tighter than the one in Ref.~\cite{Foss-FeigG}, while also using a much more intuitive approach.
In addition, our bound works for all times, unlike the bound in Ref.~\cite{Foss-FeigG}, which applies only in the long-time limit.

\begin{theorem}[Lieb-Robinson bound for long-range interactions]
	\label{TH_LR}
	Suppose $O_X$ is supported on a fixed subset $X$.
	For $\alpha>2D$, we have
	\begin{align}
	\mathcal{C}(T,R)\leq
	\begin{cases}
	c_{\lr} e^\alpha T R^{D-1} \xi_\alpha\left(\frac{R\alpha}{vT}\right), & \text{if } vT\geq \alpha,\\
	\tilde c_{\lr}\left(e^{vT}-1\right) \xi_\alpha(R), & \text{if }
	vT< \alpha.
	\end{cases} \label{EQ_OurLR}
	\end{align}
	Here $R=\dist{X,Y}$ is the distance between the supports of $O_X$ and $O_Y$, $c_{\lr},\tilde c_{\lr},v$ are constants that may depend only on $D$ [defined in \cref{APP_LR_PROOF}], and $\xi_\alpha$ is given by \cref{EQ_xi}.
\end{theorem}

Before we prove \Cref{TH_LR}, let us analyze the features of the bound.
Although the general bound in \cref{EQ_OurLR} looks complicated, it can be greatly simplified in some limits of interest.
For example, for finite $\alpha$, in the limit of large $vT>\alpha$ and large $R$ such that $R/(vT)\gg \alpha$, the term algebraically decaying with $R/(vT)$ in $\xi_\alpha(R\alpha/(vT))$ dominates the exponentially decaying one [see also \cref{EQ_xi} and \cref{EQ_xia_gtr}]. 
Therefore, the Lieb-Robinson bound in this limit takes the form:
\begin{align}
\mathcal{C}(T,R)\leq  c_{\lr,\alpha}   
\frac{T^{\alpha-D}}{R^{\alpha-2D}}, 
\end{align}
where $c_{\lr,\alpha}$ is finite and may depend on $\alpha$ [\cref{EQ_clra}].
We can immediately deduce the effective light cone given by our bound for a finite $\alpha$:
\begin{align} 
T \gtrsim R^{\frac{\alpha-2D}{\alpha-D}},\label{EQ_LightCone}
\end{align}
which is tighter than \cref{EQ_MikeLR} (as given by Ref.~\cite{Foss-FeigG}).
In particular, for $\alpha$ close to $2D$, the exponent in \cref{EQ_LightCone} can be almost twice that of Ref.~\cite{Foss-FeigG} (the larger the exponent, the tighter the light cone).

On the other hand, in the limit $\alpha\rightarrow \infty$, $vT$ is finite and therefore always less than $\alpha$.
Hence our bound converges to the short-range bound $\mathcal{C}(T,R)\leq2\tilde c_{lr}e^{vT-\gamma R}$.
We note that in this limit, the exponent of the light cone in \cref{EQ_LightCone} also converges to one, which corresponds to a linear light cone, at a linear convergence rate [see \cref{EQ_converge_rate} for details]. 
These behaviors are naturally expected since a power-law decaying interaction with very large $\alpha$ is essentially a short-range interaction.

As mentioned earlier, we derive \Cref{TH_LR} by constructing a unitary $\tilde U$ such that \emph{(i)} $\tilde U^\dag O_X \tilde U$ approximates $\Udag{\Lambda}{0,T}O_X\U{\Lambda}{0,T}$ and \emph{(ii)} $\tilde U$ commutes with $O_Y$.
We note that $\tilde U$ does not necessarily approximate $\U{\Lambda}{0,T}$.
It then follows from the two requirements that the commutator norm $\mathcal C(T,R)$, defined in \cref{EQ_CTR}, is upper bounded by the error of the approximation in (i). 
\begin{figure}[t]
	\includegraphics[width=0.45\textwidth]{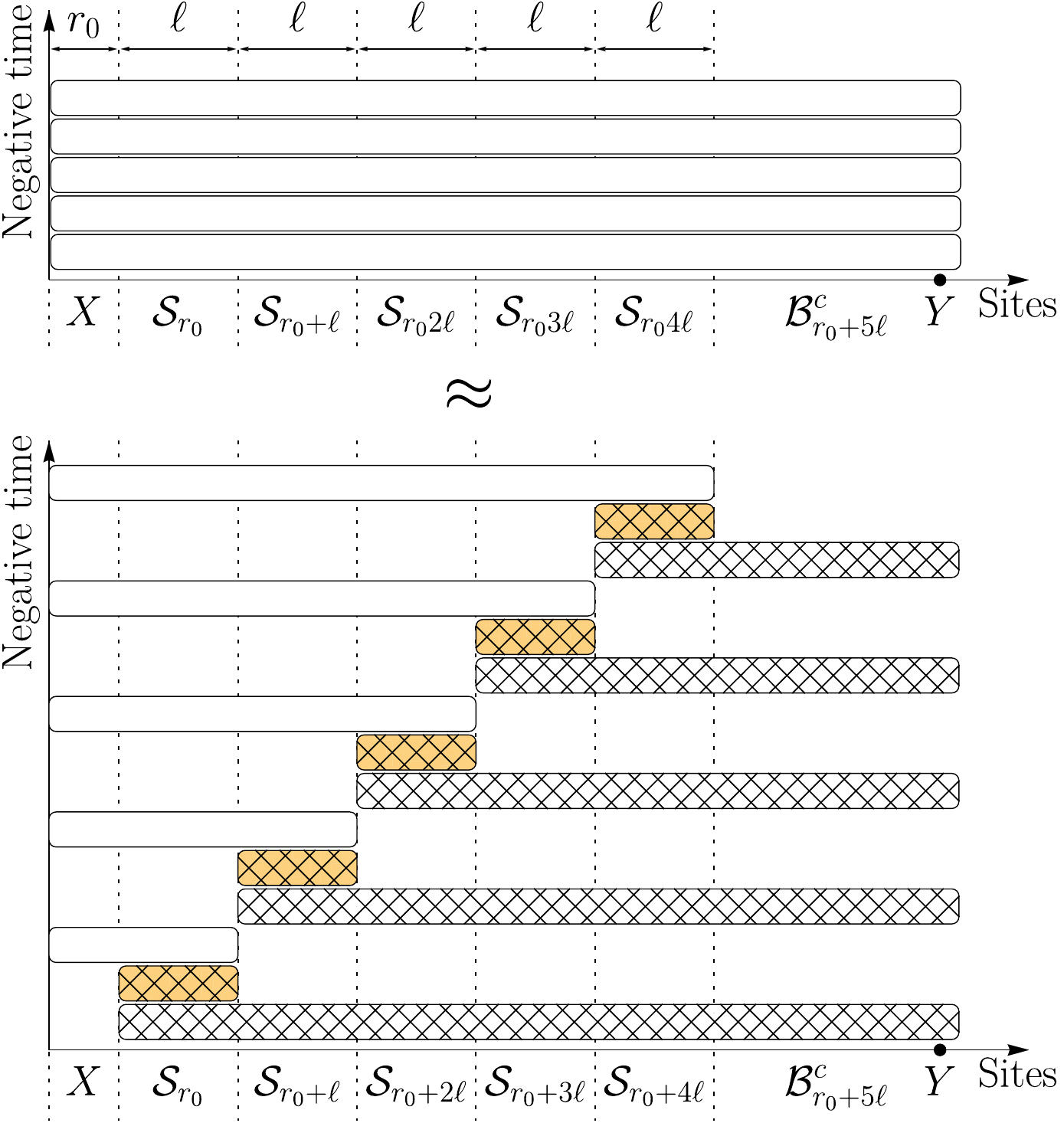}
	\caption{A construction of the unitary $\tilde U$ which results in an improved Lieb-Robinson bound for long-range interactions in Theorem~\ref{TH_LR}.
		The horizontal axes list the sites in each subset.
		Here $\mathcal B_r$ denotes a $D$-ball of radius $r$ centered on $X$, and $\mathcal S_r=\mathcal B_{r+\ell}\setminus\mathcal B_{r}$ a $D$-dimensional shell of inner radius $r$ and outer radius $r+\ell$, for some parameter $\ell$ to be chosen later. (See \cref{FIG_TH-LR-shell} in \Cref{APP_LR_PROOF} for an illustration of the sets.)
		The evolution unitaries are represented by boxes with the same color convention as in Fig.~\ref{FIG_Lem1-demo}.
		We first divide the interval $[0,T]$ into $M=5$ equal time slices (upper panel).
		Note that because we consider $O_X(T)$ in the Heisenberg picture, the vertical axis is therefore backward in time so that the bottom time slice will correspond to the first unitary applied on $O_X$.
		The evolution during each time slice is approximated by three evolutions of subsystems using Lemma~\ref{LEM_BREAK_D} (lower panel).
		The bottom two unitaries have their supports outside $X$ and therefore commute with $O_X$.
		They cancel with their Hermitian conjugates from $\tilde U^\dag$ in $\tilde U^\dag O_X \tilde U$.
		Repeating the argument for higher time slices, we can eliminate some unitaries (hatched boxes) from the construction of $\tilde U$.
		Finally, we are left with $\tilde U$ consisting only of unitaries (white boxes) that are supported entirely on the $D$-ball $\B_{r_0+5\ell}$ of radius $r_0+5\ell$.
		Therefore, $\tilde U$ commutes with $O_Y$, whose support lies in the complement $\B_{r_0+5\ell}^c$ of $\B_{r_0+5\ell}$.
	}
	\label{FIG_TH-LR-demo}
\end{figure}

We also note that the assumption on the norms of the interactions being bounded excludes several physical systems whose local dimensions are unbounded, e.g.\ bosons [see Ref.~\cite{Eisert2009,Junemann2013} for discussions of information propagation and Lieb-Robinson bounds in these systems]. 
However, our Lieb-Robinson bound may still apply if the dynamics of the systems can be restricted to local Hilbert subspaces which are finite-dimensional.
Examples of such situations include trapped ions in the perturbative regime~\cite{Kim2011} and noninteracting bosons~\cite{Deshpande17}.

To construct $\tilde U$, we use \Cref{LEM_BREAK_D} to decompose the unitary $\U{\Lambda}{0,T}$ into unitaries supported on subsystems, each of which either contains $X$ or is disjoint from $X$. 
The unitaries of the latter type can be commuted through $O_X$ to cancel out with their Hermitian conjugates from $\Udag{\Lambda}{0,T}$.
The remaining unitaries form $\tilde U$, which is supported on a smaller subset than $\U{\Lambda}{0,T}$.
In particular, with a suitable decomposition, the support of $\tilde U$ can be made to not contain $Y$, and, therefore, $\tilde U$ commutes with $O_Y$.
The step-by-step construction of the unitary $\tilde U$ has also been briefly described earlier in \cref{Sec_Overview} and in \cref{FIG_TH-LR-demo-more-step}, using the specific case of a one-dimensional system with a finite $\alpha$.
This construction immediately generalizes to higher dimensions and to arbitrary $\alpha$, including the $\alpha\to\infty$ limit.
The construction of $\tilde U$ for arbitrary $D$ is summarized in \cref{FIG_TH-LR-demo}.

We note that there is more than one way to decompose the unitary $\U{\Lambda}{0,T}$ in the construction of $\tilde U$.
Different constructions of $\tilde U$ result in different approximation errors, each of which provides a valid bound on the commutator norm $\mathcal C(T,R)$.
Therefore, the goal is to find a construction of $\tilde U$ with the least approximation error.
In \cref{APP_LR_PROOF}, we present the construction that results in the bound in \Cref{TH_LR}.
Although we have evidence suggesting that the construction is optimal, we do not rule out the existence of a better construction.

\Section{Better performance of digital simulation}
\label{Sec_qu-sim}
In this section, we generalize the algorithm in Ref.~\cite{Haah} to simulating long-range interactions.
In general, the aim of quantum simulation algorithms is to approximate the time evolution unitary $\U{\Lambda}{0,T}$ using the fewest number of primitive, e.g.\ two-qubit, quantum gates.
Here, we show that in addition to the stronger Lieb-Robinson bound presented in the previous section, Lemma~\ref{LEM_BREAK_D} can also be used to perform error analysis for the HHKL algorithm (Ref.~\cite{Haah}) in the case of interactions that decay as a power law,
therefore improving the theoretical gate count of digital quantum simulation for such interactions.

Using the best known rigorous error bounds, simulations based on the first-order Suzuki-Trotter product formula~\cite{Lloyd1073} use $\O{T^2n^6/\epsilon}$ gates to simulate the evolution $\U{\Lambda}{0,T}$ of a time-dependent Hamiltonian on $n$ sites up to a fixed error $\epsilon$.
(In this section, the big $\mathcal O$ is with respect to $n,T,$ and $1/\epsilon$.)
The generalized ($2k$)th-order product formula uses $\O{n^2(Tn^2)^{1+1/(2k)}/\epsilon^{1/2k}}$ quantum gates. 
While this scaling asymptotically approaches $\O{Tn^4}$ as $k\rightarrow\infty$, it suffers from an exponential prefactor of $5^{2k}$~\cite{BerryACS07}.
More advanced algorithms, e.g., those using quantum signal processing (QSP)~\cite{LowC2017} or linear combinations of unitaries (LCU)~\cite{BerryCCKS2015}, can reduce the gate complexity to $\O{Tn^3\log(Tn/\epsilon)}$.
Our error analysis below (\cref{LEM_HHKL_existence}) reveals that, when $\alpha$ is large, the number of quantum gates required by the HHKL algorithm to simulate long-range interactions scales better as a function of the system size than previous algorithms.

The HHKL algorithm itself uses either the QSP algorithm or the LCU algorithm as a subroutine to simulate the dynamics of a subset of the sites for one time step.  Although the QSP algorithm does not handle time-dependent Hamiltonians, LCU can be applied to time-dependent Hamiltonians.
Our results assume that \emph{(i)} the local terms $h_{\veci,\veci}(t)$ have bounded norms for all $\veci\in \Lambda$, and \emph{(ii)} the Hamiltonian $H_\Lambda(t)$ varies slowly and smoothly with time so that $h_{\abs{X}}'\equiv \max_{t}\Norm{\partial H^X_t /\partial t}$ exists and scales at most polynomially with $\abs{X}$ for all subsets $X\subset \Lambda$.  
These restrictions allow portions of the system to be faithfully simulated using LCU (or QSP, for a time-independent Hamiltonian).

\begin{figure}[t]
	\includegraphics[width=0.45\textwidth]{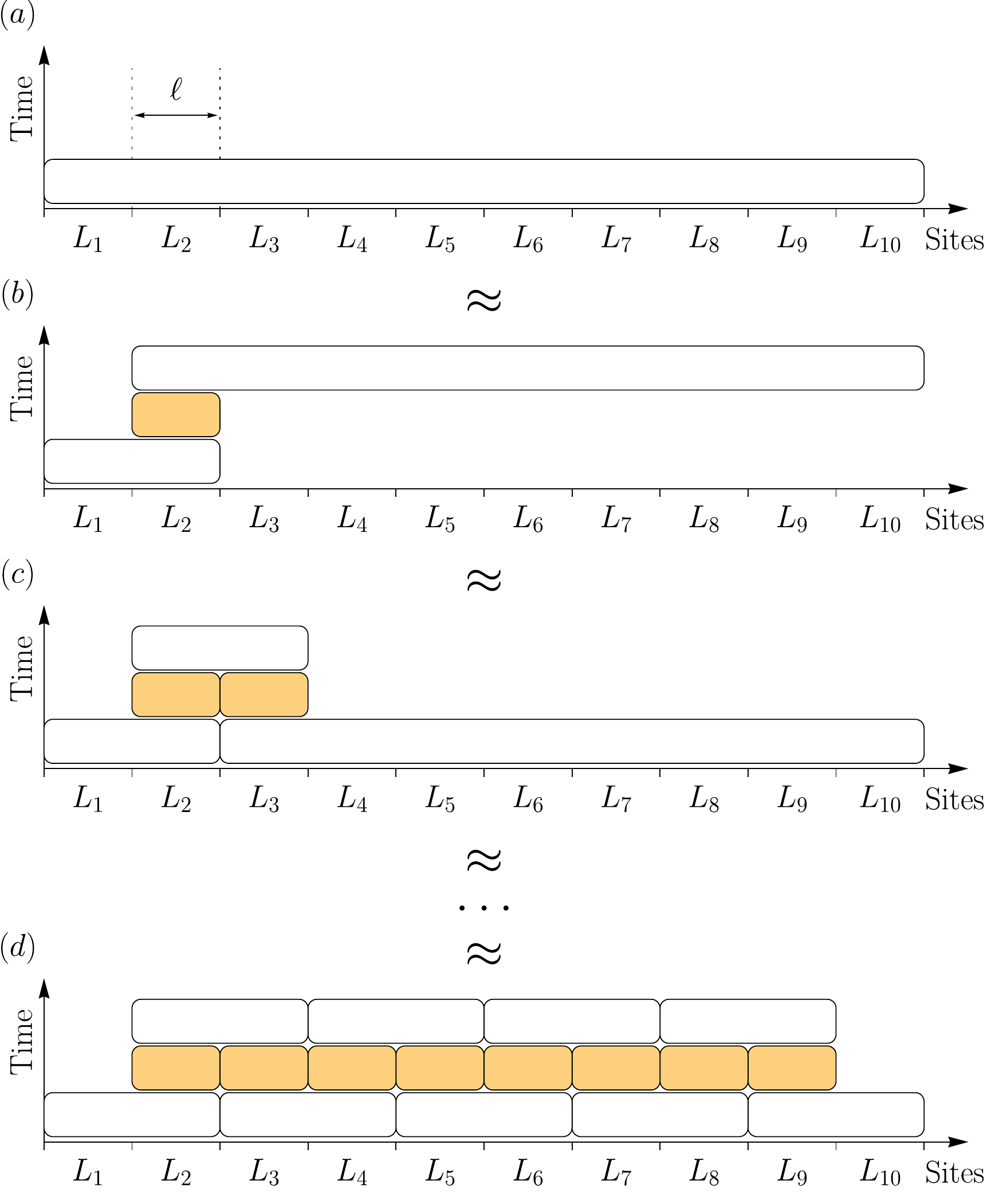}
	\caption{A demonstration of the HHKL decomposition~\cite{Haah} of the evolution of a fixed time interval for a system with $m=10$ blocks, each consisting of $\ell$ sites. 
		As before, each box represents a unitary (white) or its Hermitian conjugate (orange) supported on the covered sites. 
		Using Lemma~\ref{LEM_BREAK_D}, the HHKL decomposition approximates the evolution of the whole system [panel (a)] by three unitaries supported on subsystems [panel (b)]. 
		By applying Lemma~\ref{LEM_BREAK_D} repeatedly [panels (c) and (d)], the evolution of the whole system is decomposed into a series of evolutions of subsystems, each of size at most $2\ell$.
	}
	\label{FIG_ALGO}
\end{figure}

\subsection{HHKL-type algorithm for simulating long-range interactions}

Although Ref.~\cite{Haah} focused on simulating short-range interactions, their (HHKL) algorithm can also be used to simulate long-range interactions. 
Here, we analyze the performance of their algorithm in simulating such systems.
In the HHKL algorithm~\cite{Haah}, the evolution of the whole system is decomposed, using Lemma~\ref{LEM_BREAK_D}, into \emph{elementary} unitaries, each evolving a subsystem of at most $(2\ell)^D$ sites, where $\ell$ is again a length scale to be chosen later.
For a fixed time $t$, the algorithm simply simulates each of these elementary unitaries using one of the existing quantum simulation algorithms.
In particular, we shall use LCU or (for a time-independent Hamiltonian) QSP  due to their logarithmic dependence on the accuracy.

In this section, we assume $\alpha$ is finite and analyze the gate count in the limit of large system size $n\gg \alpha$. 
As a consequence, the block size $\ell$ can also be taken to be much larger than $\alpha$.
For simplicity, we will not keep track of constants that may depend on $\alpha$.
Recall that in this limit, the error bound in \Cref{LEM_BREAK_D} is at most
\begin{align}
\O{\frac{\Phi(A)}{\ell^{\alpha-D-1}}},\label{EQ_error_fixed_a}
\end{align}
where we have assumed $t=\O{1}$.
Using \cref{LEM_BREAK_D}, we obtain the error bound for the first step of the HHKL algorithm, which can be summarized by the following lemma.

\begin{lemma}[HHKL decomposition]
	\label{LEM_HHKL_existence}
	There exists a circuit that approximates $\U{\Lambda}{0,T}$ up to error $\O{{Tn}/{\ell^{\alpha-D}}}$, where $\ell\leq n^{1/D}/2$ is a free parameter. The circuit has depth at most $3^D T$ and consists of $\O{{Tn}/{\ell^D}}$ \emph{elementary} unitaries, each of which evolves a subsystem supported on at most $(2\ell)^D$ sites for time $t = \O{1}$.   
\end{lemma}
\begin{proof}
	We now demonstrate the proof by constructing the circuit for a one-dimensional lattice (Fig.~\ref{FIG_ALGO}). A generalization of the proof to arbitrary dimension follows the same lines and is presented in \Cref{SM_Sec_HHKL}.
	
	First, we consider $M \propto T$ equal time intervals $0=t_0<t_1<\dots<t_M=T$ such that $t_{j+1}-t_{j}= t=T/M$ is a constant for all $j=0,\dots,M-1$. 
	The simulation of $U^\Lambda_{0,T}$ then naturally decomposes into $M$ consecutive simulations of $U^\Lambda_{t_j,t_{j+1}}$.
	We then divide the system into $m$ consecutive disjoint blocks, each of size $\ell=n/m$ (Fig.~\ref{FIG_ALGO}).
	Denote by $L_k~(k=1,\dots,m)$ the set of sites in the $k$-th block.
	Using Lemma~\ref{LEM_BREAK_D}, we can approximate
	\begin{align}
	\U{\Lambda}{0,t}&\approx \U{L_1\cup L_2}{0,t}\left(\U{L_2}{0,t}\right)^\dag\U{L_2\cup L_3\cup\dots\cup L_m}{0,t}.
	\end{align}
	This approximation can be visualized using the top two panels of Fig.~\ref{FIG_ALGO}.
	Repeated application of Lemma~\ref{LEM_BREAK_D} yields the desired circuit (bottom panel of Fig.\ 4), with each elementary unitary evolving at most $2\ell$ sites for time $t$.
	
	To obtain the error estimate in Lemma~\ref{LEM_HHKL_existence}, we count the number of times Lemma~\ref{LEM_BREAK_D} is used in our approximation. 
	In each of the $M$ time slices, we use the lemma $\O{m}=\O{{n}/{\ell}}$ times, each of which contributes an error of $\O{{1}/{\ell^{\alpha-2}}}$ [see \cref{EQ_error_fixed_a} with $\Phi=\O{1}$ in one dimension].
	Therefore, with $M\propto T$, the error of using the constructed circuit to simulate $\U{\Lambda}{0,T}$ is 
	\begin{align}
	\O{M\frac{n}{\ell}\frac{1}{\ell^{\alpha-2}}} = \O{\frac{Tn}{\ell^{\alpha-1}}},\label{EQ_error_scaling_ell}
	\end{align}
	as given in Lemma~\ref{LEM_HHKL_existence}.  
\end{proof}
The error bound for the approximation in \cref{LEM_HHKL_existence} leads to an upper bound on the gate complexity of digital quantum simulation, as stated in the following theorem.
\begin{theorem}
	\label{TH_Statement}
	For $\alpha>2D$, there exists a quantum algorithm for simulating $\U{\Lambda}{0,T}$ up to error at most $\epsilon$ with gate complexity
	\begin{align}
	G_D = \O{Tn\left(\frac{Tn}{\epsilon}\right)^{\frac{2D}{\alpha-D}}\log\frac{Tn}{\epsilon}}.\label{EQ_GD}
	\end{align}
\end{theorem}

This gate complexity can be achieved by applying the HHKL algorithm~\cite{Haah} for long-range interactions, as described above.
First, the evolution of the whole system $\U{\Lambda}{0,T}$ is approximated by $\O{{Tn}/{\ell^D}}$ elementary unitaries as provided in Lemma~\ref{LEM_HHKL_existence}.
Each of these elementary unitaries is then simulated using one of the existing algorithms, e.g., LCU, with  error that we require to be at most ${\epsilon \ell^D}/{Tn}$.
If the Hamiltonian is time-independent, one can also use the QSP algorithm to simulate the elementary unitaries.

In the decomposition of the evolution, the accuracy of the approximation can be improved by increasing the block size $\ell$. 
By Lemma~\ref{LEM_HHKL_existence}, to achieve an overall error at most $\epsilon$, we need
\begin{align}
\ell \propto \left(\frac{Tn}{\epsilon}\right)^{\frac{1}{\alpha-D}}. \label{EQ_choosel}
\end{align}
When simulating the elementary unitaries, since each is an evolution of at most $(2\ell)^D$ sites for time $t=\O{1}$, the LCU algorithm with error at most ${\epsilon \ell^D}/{Tn}$ uses $\O{\ell^{3D}\log\left(\frac{Tn}{\epsilon}h_{\ell^D}'\right)}$ two-qubit gates~\cite{BerryACS07}.
Recall that we assume $h_{\ell^D}'$ scales at most polynomially with $\ell^{D}$.
With the block size $\ell$ from Eq.~\eqref{EQ_choosel}, we find the total gate complexity of simulating the $\O{{Tn}/{\ell^D}}$ elementary unitaries is
\begin{align}
G_D&=\O{\frac{Tn}{\ell^D}~\ell^{3D}\log\left(\frac{Tn}{\epsilon}h_{\ell^D}'\right)}\\
&=\O{Tn\left(\frac{Tn}{\epsilon}\right)^{\frac{2D}{\alpha-D}}\log{\frac{Tn}{\epsilon}}}.\label{EQ_final_count}
\end{align}

The scaling of $G_D$ as a function of the system size $n$ is significantly better than existing algorithms for large $\alpha$. In particular, at $T = n$, this HHKL algorithm for long-range interactions requires only $\O{n^{2+\frac{4D}{\alpha-D}}\log n}$ gates, while algorithms such as QSP or LCU use $\O{n^4\log n}$ gates or more. 
Therefore, the algorithm provides an improvement for $\alpha> 3D$.
However, the gate complexity of the algorithm depends polynomially on $1/\epsilon$, in contrast to the logarithmic dependence achieved by QSP and LCU, and by the HHKL algorithm for systems with short-range interactions.
While this $\text{poly}(1/\epsilon)$ scaling is undesirable, in practice, the total error of the simulation is often set to a fixed constant (for example, see Ref.~\cite{Babbush15}) and effectively the dependence of $\epsilon$ only contributes a prefactor to the gate complexity of the algorithm.

As an example, in \cref{FIG_actual-gate-count}, we estimate the actual gate count of the HHKL algorithm in simulating a Heisenberg chain [\cref{EQ_Heisenberg}] and compare it with the gate count of the QSP algorithm (up to the same error tolerance).
Because of the $\textrm{poly}(1/\epsilon)$ overhead, the HHKL algorithm based on Lieb-Robinson bounds uses more quantum gates for simulating small systems, but eventually outperforms the QSP algorithm when the system size $n$ is large.

\begin{figure}[t]
\includegraphics[width=0.45\textwidth]{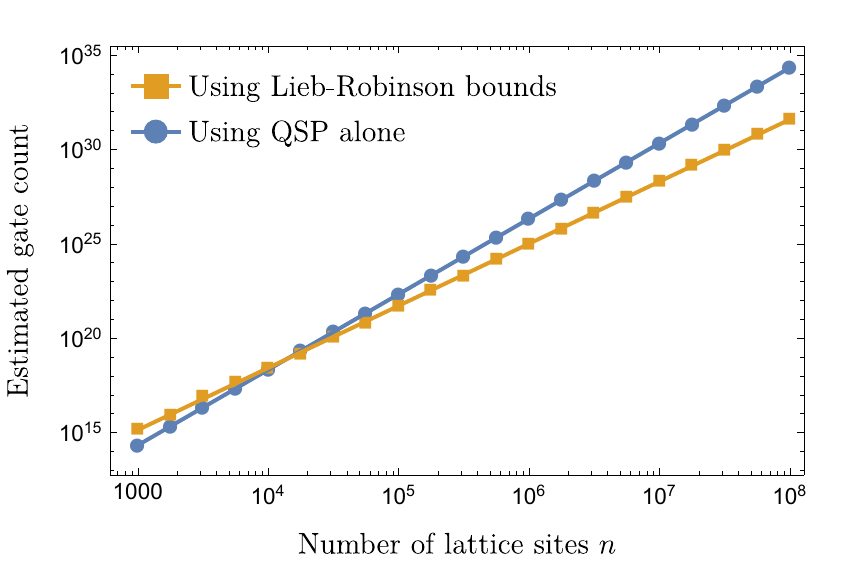}
\caption{The gate count for simulating the dynamics of a one-dimensional Heisenberg chain [\cref{EQ_Heisenberg}] of length $n$, with $\alpha = 4, T=n,$ and $ \epsilon = 10^{-3}$. 
We compare the gate count of the HHKL algorithm (orange square) to the QSP algorithm (blue circle). 
Note that the HHKL algorithm based on Lieb-Robinson bounds also uses the QSP algorithm as a subroutine.
We obtain the scatter points using the approach described in \cref{APP_gate_count} and fit them to a power-law model (solid lines).
The asymptotic scalings of the gate count obtained from the power-law fits ($n^{3.29}$ for HHKL, $n^{4.00}$ for QSP) agree well with our theoretical predictions (see \cref{TAB_gate_count}). 
}
\label{FIG_actual-gate-count}
\end{figure}

It is also worth noting that, in the limit $\alpha\rightarrow\infty$, the gate complexity becomes $\O{Tn\log\left({Tn}/{\epsilon}\right)}$, which coincides (up to a polylogarithmic factor) with the result for short-range interactions in Ref.~\cite{Haah}.
This behavior is expected, given that a power-law decaying interaction with $\alpha\rightarrow\infty$ is essentially a nearest-neighbor interaction.
However, we caution readers that at the beginning of this section, we have assumed that $\alpha$ is finite so that $n\gg \alpha$. Hence, the gate count in \cref{EQ_final_count} is technically not valid in the limit $\alpha\to\infty$.
Nevertheless, the error bound in \Cref{LEM_BREAK_D} reproduces the estimate for short-range interactions in Ref.~\cite{Haah}, and therefore, repeating the argument of this section in the limit $\alpha\to\infty$ should also reproduce the gate count for simulating short-range interactions in Ref.~\cite{Haah}.


\subsection{Numerical evidence of potential improvement}
\label{Subsec_num}
Up to now, we have seen that Lieb-Robinson bounds can improve the error bounds of quantum simulation algorithms, as demonstrated by the HHKL algorithm.
We now provide numerical evidence hinting at the possibility of further improving the error bounds.

Although the HHKL algorithm outperforms previous ones when $\alpha>3D$, it remains an open question whether there is a faster algorithm for simulating long-range interactions.  
We also note that the gate complexities are only theoretical upper bounds, and these algorithms may actually perform better in practice~\cite{ChildsMNRS2017}.

As an example, we compute the empirical gate count of a Suzuki-Trotter product formula simulation of a one-dimensional long-range interacting Heisenberg model
\begin{equation}
H=\sum_{i=1}^{n-1}\sum_{j=i+1}^{n} \frac{1}{|i-j|^4}\vec \sigma_i \cdot \vec \sigma_{j}+\sum_{i=1}^{n}B_i \sigma_i^z, \label{EQ_Heisenberg}
\end{equation}
where $B_j \in [-1,1]$ are chosen uniformly at random and $\vec \sigma_j = (\sigma^x_j,\sigma^y_j,\sigma^z_j)$ denotes the vector of Pauli matrices on the qubit $j$. 
Specifically, we consider a simulation using the fourth-order product formula (PF4).
We use a classical simulation to determine the algorithm's performance for systems of size $n = 4$ to $n = 12$ for time $T = n$, and extrapolate to larger systems.
For each $n$, we search for the minimum number of gates for which the simulation error is at most $\varepsilon=10^{-3}$. 
We plot in \cref{Sec_num} this empirical gate count, which appears to 
scale only as $\O{n^{3.64}}$ with the system size $n$.
We list in \cref{TAB_gate_count} the gate counts of several popular algorithms for comparison.
The \emph{theoretical} gate complexity of PF4 is $\O{n^{5.75}}$~\cite{Lloyd1073}, while the QSP and LCU algorithms both have complexity $\O{n^4\log n}$.
These numerics show that the PF4 algorithm for simulating long-range interacting systems performs better in practice than theoretically estimated; in fact, it even performs almost as well as the HHKL algorithm based on Lieb-Robinson bounds [which scales as $\O{n^{3.33}\log n}$ by our earlier analysis]. 
Whether other quantum simulation algorithms, including the HHKL algorithm, can perform better than suggested by the existing bounds remains an important open question.

\begin{table}[t]
	\begin{tabular}{l l l}
		\toprule 
		Algorithm & Scaling with $n=T$ & Scaling with $\epsilon$  \\ 
		\colrule
		Empirical PF4 & $\O{n^{3.64}}$ &--- \\ 
		Our HHKL bound & $\O{n^{3.33}\log n}$ &$\O{\log({1}/{\epsilon})/{\epsilon^{0.67}}}$ \\ 
		PF4 bound~\cite{BerryACS07}& $\O{n^{5.75}}$ & $\O{{1}/{\epsilon^{0.25}}}$ \\  
		QSP bound~\cite{LowC2017} & $\O{n^4\log n}$ & $\O{\log({1}/{\epsilon})}$ \\ 
		LCU bound~\cite{BerryCCKS2015}& $\O{n^4 \log n}$ & $\O{\log({{1}/{\epsilon}})}$ \\
		\botrule 
	\end{tabular} 
	\caption{A comparison between the gate complexity of several quantum simulation algorithms for simulating one-dimensional power-law systems at $T=n$ and $\alpha=4$.
		Our analysis shows that the HHKL algorithm performs at least as well as the empirical gate count of PF4, while having a similar $\text{poly}(1/\epsilon)$ scaling with the error $\epsilon$.  
		It is not known whether the empirical gate count of PF4 can scale with $\epsilon$ better than suggested by the best proven bound (the third row).
	}
	\label{TAB_gate_count}
\end{table}
\section{Conclusion \& outlook}
\label{Sec_Outlook}
To conclude, we derived an improved bound on how quickly quantum information propagates in systems evolving under long-range interactions. 
The bound applies to power-law interactions with $\alpha>2D$, such as dipole-dipole interactions in 1D (often realizable with nitrogen-vacancy centers~\cite{Maze2011} or polar molecules~\cite{Yan2013}), trapped ions in 1D~\cite{Britton2012,Kim2011}, and van-der-Waals-type interactions between Rydberg atoms~\cite{Saffman10} in either 1D or 2D\@. 
For finite $\alpha>2D$, our Lieb-Robinson bound gives a tighter light cone than previously known bounds\textemdash including the one used in the proof of Lemma~\ref{LEM_BREAK_D}.
As of yet, we are not aware of any physical systems that saturate the Lieb-Robinson bounds for power-law interactions, including the new bound.
In the limit $\alpha \rightarrow \infty$, our bound asymptotically approaches the exponentially decaying bound for short-range interactions. 
Our bound gives a linear light cone only in this limit, however, and it remains an open question whether there exists a stronger bound with a critical $\alpha_c$ such that the light cone is exactly linear for $\alpha \geq \alpha_c$~\cite{Luitz2019}.
Currently, there are no known methods for quantum information transfer that are faster than linear for $\alpha \geq D+1$. 
It is possible, therefore, that a stronger bound exists with a finite $\alpha_c \geq D+1$.
It is our hope that the present work, as well as the techniques that we use, will help motivate the search for such stronger bounds.

Our technique immediately extends the HHKL algorithm in Ref.~\cite{Haah} to the digital quantum simulation of the above systems.
Our error bounds indicate that the gate complexity of the algorithm is better than that of other state-of-the-art simulation algorithms when $\alpha$ is sufficiently large ($\alpha>3D$), and matches that of the short-range algorithm when $\alpha \rightarrow \infty$. 

However, the empirical scaling of other algorithms\textemdash such as product formulas\textemdash indicates that this gate complexity may only be a loose upper bound to the true quantum complexity of the problem. 
While a matching lower bound for the gate complexity of the HHKL algorithm is provided in Ref.~\cite{Haah} for Hamiltonians with short-range interactions, we do not know of any techniques that could provide a corresponding bound for long-range interactions.
In addition to improving the quantum gate complexity, our results may also aid in the design of better classical algorithms for simulating long-range interacting quantum systems.
In particular, while we still expect the classical gate complexity to be exponential in the simulation time, there may be room for a polynomial improvement. 

While the use of Lieb-Robinson bounds to improve the performance of quantum algorithms is a natural extension of Haah et al., the opposite direction\textemdash using quantum algorithms to improve Lieb-Robinson bounds\textemdash is new.
The connection from quantum simulation algorithms to Lieb-Robinson bounds that we have established opens another avenue for the condensed matter and atomic/molecular/optical physics communities to potentially benefit from future advances in quantum algorithms.
In addition to proving a stronger Lieb-Robinson bound, the tools we developed may help to answer other questions regarding both short-range and long-range interacting systems. 
Using the same unitary construction as \Cref{TH_LR}, we can generalize the bounds on connected correlators~\cite{Bravyi06,Tran17} to long-range interacting systems.
Our results can also provide a framework for proving tighter bounds on higher-order commutators, such as out-of-time-order correlators~\cite{Larkin1969,kitaev_soft_2018}.
Previous methods used to derive Lieb-Robinson bounds\textemdash due to their use of the triangle inequality early in their proofs\textemdash have not been able to capture the nuances in the growth of such correlators.  
In addition to the more intuitive proof of the Lieb-Robinson bounds, our framework can be used to provide an alternative, simpler proof of the classical complexity of the boson-sampling problem~\cite{Aaronson2011}, which generalizes the result in Ref.~\cite{Deshpande17} to long-range interactions and also to more general Hamiltonians with arbitrary local interactions~\cite{Boson-sampling-TBP}.
By taking advantage of the unitary decomposition in \cref{LEM_BREAK_D}, we obtain a longer time interval within which the sampler in Ref.~\cite{Deshpande17} is efficient~\cite{Boson-sampling-TBP}.
Moreover, by generalizing from two-body to many-body interactions, our technique may find applications in systems whose Hamiltonians include interaction terms between three or more sites, e.g.\ many-body localized systems in the $l$-bit basis~\cite{Serbyn14}.

\textit{Note added:}
Shortly after we submitted our paper to arXiv,  Else~et~al.~\cite{Else18} posted their work on a different Lieb-Robinson bound for power-law decaying interactions. 
For Hamiltonians consisting of at most two-body interactions, the bound in Ref.~\cite{Else18} and our bound both apply in the same regime, $\alpha>2D$~\footnote{
Note that there is a difference between the definition of the exponent~$\alpha$ in Ref.~\cite{Else18} and in this paper. 
For two-body Hamiltonians, the exponent~$\alpha$ in Ref.~\cite{Else18} is equal to our $\alpha$ minus the dimension $D$}. 
Within this regime, our bound results in a strictly tighter light cone than Ref.~\cite{Else18}.
However, the bound in Ref.~\cite{Else18} also applies to Hamiltonians consisting of $k$-body interactions, for any integer $k$.
Generalizing our framework to cover such $k$-body interactions would be an interesting future direction.

\begin{acknowledgments}
	We thank G. H. Low, A. Deshpande, P. Titum, T. Zhou, and Z.-X. Gong for helpful discussions. 
	MCT, AYG, ZE and AVG acknowledge funding from the U.S.\ Department of Energy ASCR Quantum Testbed Pathfinder program (Award No. DE-SC0019040), ARO MURI, ARO, NSF Ideas Lab on Quantum Computing, ARL CDQI, NSF PFC at JQI, AFOSR, and the U.S. Department of Energy BES Materials and Chemical Sciences Research for Quantum Information Science program (Award No. DE-SC0019449). 
	AMC and YS acknowledge funding from ARO MURI, CIFAR, NSF, and the U.S.\ Department of Energy, Office of Science, Office of Advanced Scientific Computing Research, Quantum Algorithms Teams and Quantum Testbed Pathfinder programs (Award No. DE-SC0019040).
	MCT is supported in part by the NSF Grant No.~NSF PHY-1748958 and the Heising-Simons Foundation.
	AYG is supported by the NSF Graduate Research Fellowship Program under Grant No.\ DGE 1322106.
    JRG is supported by the NIST NRC Research Postdoctoral Associateship Award and performed his work in part at the Aspen Center for Physics, which is supported by National Science Foundation grant PHY-1607611.
    ZE is supported in part by the ARCS Foundation.
\end{acknowledgments}

\appendix 
\crefalias{section}{appsec}
\crefalias{subsection}{appsec}


\Section{Evaluations of the sum in Lemma~\ref{LEM_BREAK_D}}
\label{APP_LEM_1_PROOF}
In this section, we shall show
how we bound $\dtrunc$ from \cref{EQ_sumAC} (\Cref{SM_Subsec_dtrunc}) and $\doverlap$ from \cref{EQ_doverlap} (\Cref{SM_Subsec_doverlap}) in the proof of Lemma~\ref{LEM_BREAK_D}.

\subsection{Evaluation of $\dtrunc$}
\label{SM_Subsec_dtrunc}
In this subsection, we provide explicit calculations of $\dtrunc$ in Eq.~\eqref{EQ_sumAC}.
Recall that $\ell = \dist{A,C}$ is the shortest distance between any two points in $A$ and $C$.
Therefore, $\norm{\vec a-\vec c}$ is always greater than $\ell$.
For each $\vec a\in A$, let $\ell_{\vec a}=\dist{\vec a, C}$ be the minimum distance from $\vec a$ to the set $C$ and $C_{\vec a}=\left\{\veci\in \Lambda: \dist{\vec a,\veci}\geq \ell_{\vec a} \right\}$.
Clearly, $C$ is a subset of $C_{\vec a}$.
Therefore, 
\begin{align}
\norm{\dtrunc} &= \norm{H_{A:C}}\leq \sum_{\vec a\in A}\sum_{\vec c \in C} \frac{1}{\norm{\vec a-\vec c}^{\alpha}}\\
&\leq \sum_{\vec a\in A}\sum_{\vec c \in C_{\vec a}} \frac{1}{\norm{\vec a-\vec c}^{\alpha}}
=\sum_{\vec a\in A}\sum_{\substack{\vec r\\\norm{\vec r}\geq \ell_{\vec a}}} \frac{1}{\norm{\vec r}^{\alpha}}\label{EQ_sumAC_1}\\
&\leq \lambda_1\sum_{\vec a\in A}\frac{1}{(\ell_{\vec a}-\sqrt{D})^{\alpha-D}}\label{EQ_sumveca},
\end{align} 
where $\lambda_1$ is a constant independent of $\vec a$ and the sum over $\vec r$ is bounded using Lemma~\ref{LEM_SIMPLE_SUM} in \Cref{Sec_math}.

Next, to evaluate the sum over $\vec a$, we parameterize the sites in the set $ A$ by their distance to its boundary $\partial A$. 
Note that by assumption the interior of $A$ is non-empty, so that $A \neq \partial A$.
Roughly speaking, there will be at most $\O{\Phi( A)}$ sites whose distances to the boundary $\partial A$ is between $\ell$ and $\ell+\mu$, for each $\mu = 0,1,\dots$, where $\Phi( A)$ is the boundary area of $ A$.
Therefore, we have (see \Cref{LEM_PARA} in \Cref{Subsec_para} for a rigorous proof) 
\begin{align}
\norm{\dtrunc }
&\leq 2\eta\lambda_1 \Phi(A)\sum_{\mu=0}^\infty \frac{1}{(\ell+\mu-\sqrt{D})^{\alpha-D}}\\
&\leq 2\eta\lambda_1\lambda_{2}2^{\alpha-D} \frac{\Phi( A)}{\ell^{\alpha-D-1}}
=c_{\tr}2^\alpha \frac{\Phi( A)}{\ell^{\alpha-D-1}},\label{EQ_ctr_def}
\end{align}
for $\ell>2\sqrt{D}$, 
where $\lambda_{2}$ is a constant that arises after using \cref{LEM_SIMPLE_SUM} to bound the sum, and the factor $2^{\alpha-D}$ is because we lower bound $\ell-\sqrt D\geq \ell/2$ to simplify the expression.
The constants are later absorbed into the definition of $c_{\tr}$.
\subsection{Evaluation of $\doverlap$}
\label{SM_Subsec_doverlap}
In this section, we show how we bound $\doverlap$ from \cref{EQ_doverlap} in the proof of \Cref{LEM_BREAK_D}.
To estimate $\doverlap$, we use the following lemma, which generalizes a similar lemma in Ref.~\cite{Haah} to arbitrary, time-dependent Hamiltonians.
\begin{lemma}
	\label{APP_LEM_EVO_APPROX}
	Let $\Omega\subset \Lambda$ be a subset of sites. 
	Let $H_\Omega(t)=\sum_{i,j\in\Omega}h_{ij}(t)$ be the terms of $H_\Lambda(t)$ supported entirely on $\Omega$.
	Let $O_X(\tau)$ be an observable supported on a subset $X$ at a fixed time $\tau$.
	We have:
	\begin{align}
	&\Norm{\Udag{\Lambda}{0,t}O_X(\tau)\U{\Lambda}{0,t}-\Udag{\Omega}{0,t}O_X(\tau)\U{\Omega}{0,t}}\nonumber\\
	&\leq \int_0^t ds~ \Norm{\comm{\Udag{\Omega}{s,t}O_X(\tau) \U{\Omega}{s,t},H_\Lambda(s)-H_\Omega(s)}},
	\end{align}
	where $\U{\Lambda}{0,t} =\mathcal{T}\exp\left(-i\int_0^t H_\Lambda(s)ds\right)$.
\end{lemma}
\begin{proof}
	To prove the lemma, we shall move into the interaction picture of $H_\Omega(t)$ and treat $V(t)\equiv H_\Lambda(t)-H_{\Omega}(t)$ as a perturbation.
	Let $V_I(t) = \Udag{\Omega}{0,t}V(t)\U{\Omega}{0,t}$ and $U_I(t) = \mathcal T \exp\left(-i\int_0^t V_I(s)ds\right)$ be respectively the Hamiltonian and the evolution operator in the interaction picture.
	We have:
	\begin{align}
	&\Norm{\Udag{\Lambda}{0,t}O_X(\tau)\U{\Lambda}{0,t} - \Udag{\Omega}{0,t}O_X(\tau)\U{\Omega}{0,t}} \nonumber\\
	&=\Norm{ \int_0^t ds~ \frac{d}{ds}\left(U^\dag_I(s) \Udag{\Omega}{0,t}O_X(\tau)\U{\Omega}{0,t} U_I(s)\right)} \\
	&= \Norm{\int_0^t ds~  U^\dag_I(s) \comm{\Udag{\Omega}{0,t}O_X(\tau)\U{\Omega}{0,t},V_I(s)}U_I(s)}\\
	&\leq\int_0^t ds~ \Norm { \comm{\Udag{\Omega}{0,t}O_X(\tau)\U{\Omega}{0,t},V_I(s)}}	\\
	&\leq\int_0^t ds~ \Norm { \comm{\Udag{\Omega}{s,t}O_X(\tau)\U{\Omega}{s,t},V(s)}}.		
	\end{align}
	Thus, Lemma~\ref{APP_LEM_EVO_APPROX} follows.
\end{proof}
By substituting $\Lambda \rightarrow AB,$ $\Omega \rightarrow B$, $O_X \rightarrow H_{B:C}$, $\tau\rightarrow t$, and noting that operators supported on disjoint subsets commute, we can show using Lemma~\ref{APP_LEM_EVO_APPROX} that
\begin{align}
&\Norm{\doverlap}=\Norm{\Udag{AB}{0,t}H_{B:C}(t)\U{AB}{0,t}-\Udag{B}{0,t}H_{B:C}(t)\U{B}{0,t}}\nonumber\\
&\leq \int_0^t ds \Norm{\comm{\Udag{B}{s,t}H_{B:C}(t) \U{B}{s,t},H_{A:B}(s)}}\nonumber\\
&\leq \sum_{\vec a\in A}\sum_{\vec b,\vec b'\in B}\sum_{\vec c\in C}\int_0^t ds \Norm{\comm{\Udag{B}{s,t}h_{\vec b,\vec c}(t) \U{B}{s,t},h_{\vec a,\vec b'}(s)}},\label{EQ_sum_ABC}
\end{align}
where the observables are, in general, evaluated at different times $s\leq t$.
We note that while it is necessary to keep track of $s,t$ for completeness, one should pay more attention to the supports of the operators, as they carry useful information about the locality of the system.

In fact, let us pause for a moment to discuss why the right hand side of Eq.~\eqref{EQ_sum_ABC} should be small when $\ell$, the distance between $A$ and $C$, is large. 
Whenever the supports of $h_{\vec a,\vec b'}(s)$ and $h_{\vec b,\vec c}(t)$ are far from each other, we can bound their commutator norm using a Lieb-Robinson bound for long-range interactions.
We use the bound by Gong \etal~\cite{GongFF}:
\begin{align}
&\Norm{\comm{\Udag{B}{s,t}h_{\vec b,\vec c}(t) \U{B}{s,t},h_{\vec a,\vec b'}(s)}}\nonumber\\
&\leq ce^{v(t-s)}\Norm{h_{\vec b,\vec c}(t)}\Norm{h_{\vec a,\vec b'}(s)} \left(\frac{1}{(1-\gamma)^\alpha}\frac{1}{r^\alpha}+\frac{1}{e^{\gamma r}}\right),\label{EQ_LR_Gong}
\end{align}
where $r = \dist{\{\vec a,\vec b'\},\{\vec b,\vec c\}}$ is the distance between the supports, $\gamma\in(0,1)$ is a constant that can be made arbitrarily close to 1, and $c,v$ are constants that depend only on $D$.

\begin{figure}
	\includegraphics[width=0.45\textwidth]{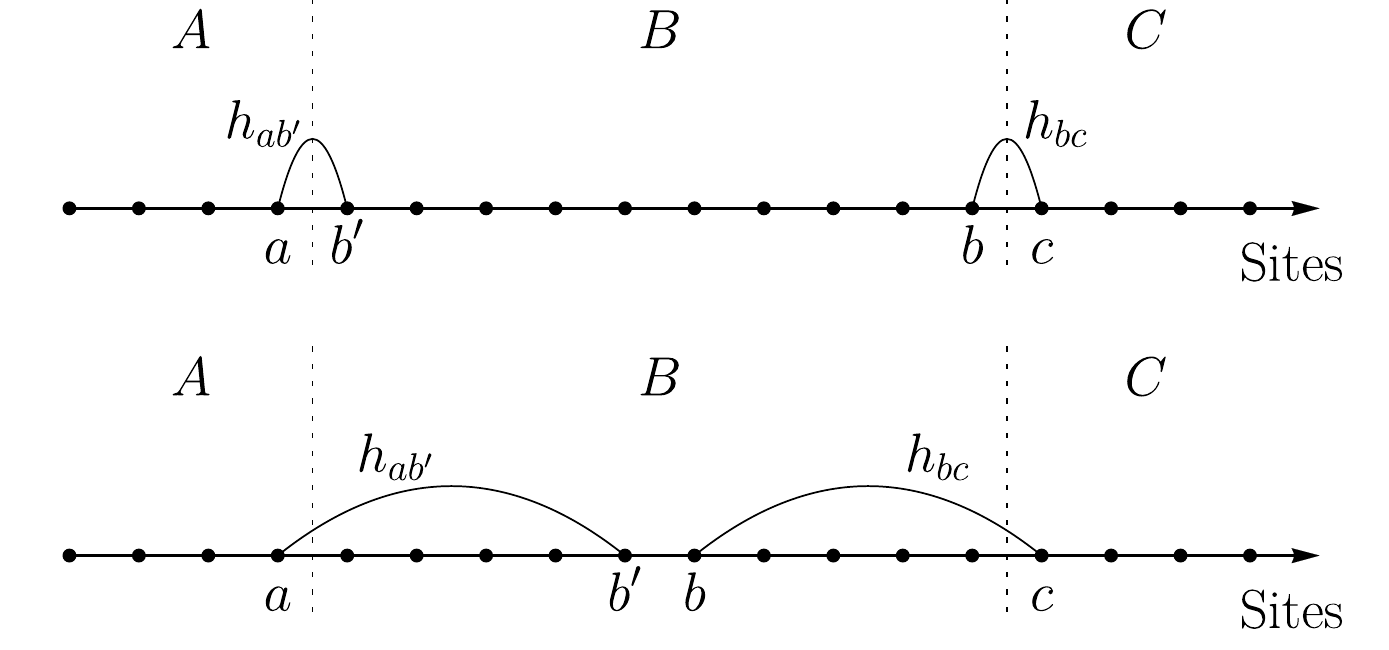}
	\caption{An illustration of $h_{ab'}$ and $h_{bc}$ in a one-dimensional lattice.
		For short-range interactions, the sets $\{a,b'\}$ and $\{b,c\}$ are separated by a distance of the same order as the size of $B$ (upper figure). 
		The contributions from these terms to $\doverlap$ are bounded using a Lieb-Robinson bound.
		However, for long-range interactions, $\{a,b'\}$ and $\{b,c\}$ can be geometrically close to each other (lower figure).
		In such cases, the norms of $h_{ab'}$ and $h_{bc}$ decay as $\Abs{b'-a}^{-\alpha}$ and $\Abs{c-b}^{-\alpha}$ and, therefore, their contributions to $\doverlap$ are small.
	}
	\label{FIG_doverlap}
\end{figure}
However, in contrast to short-range interacting systems, here $\vec b,\vec b'$ run over all possible sites in $B$, so in principle the distance between the supports of $h_{\vec a,\vec b'}(s)$ and $h_{\vec b,\vec c}(t)$ can be small (Fig.~\ref{FIG_doverlap}). 
Fortunately, if that is indeed the case, then although the Lieb-Robinson bound is trivial, the assumption that $\norm{h_{\vec a,\vec b'}}$ and $\norm{h_{\vec b,\vec c}}$ fall off as ${\norm{\vec b'-\vec a}^{-\alpha}}$ and ${\norm{\vec c-\vec b}^{-\alpha}}$, respectively, makes the summand in Eq.~\eqref{EQ_sum_ABC} small.

Let us now evaluate the sum in \cref{EQ_sum_ABC}.
In the following, we shall consider $\vec b\neq \vec b'$, since the estimation for the case $\vec b = \vec b'$ follows a similar, but less complicated argument.
Using Gong \etal's Lieb-Robinson bound for long-range interactions~\cite{GongFF}:
\begin{align}
&\int_0^t ds\norm{\comm{\Udag{B}{s,t}h_{\vec b,\vec c}(t) \U{B}{s,t},h_{\vec a,\vec b'}(s)}}\nonumber\\
&\leq c\int_0^t ds\norm{h_{\vec b,\vec c}(t)}\norm{h_{\vec a,\vec b'}(s)} \left(\frac{e^{v(t-s)}}{((1-\gamma)r)^\alpha}+\frac{e^{v(t-s)}}{e^{\gamma r}}\right)\nonumber\\
&\leq \frac c v \frac{(e^{vt}-1)}{\norm{\vec b-\vec c}^\alpha\norm{\vec a - \vec b'}^\alpha} \left(\frac{1}{(1-\gamma)^\alpha}\frac{1}{r^\alpha}+\frac{1}{e^{\gamma r}}\right),\label{EQ_doverlap_two_parts}
\end{align}
where $\gamma\in(0,1)$ is a constant that can be chosen arbitrarily close to 1, while $c,v$ are finite and bounded constants for all $\alpha$, and
\begin{align}
r&=\dist{\{\vec a,\vec b'\},\{\vec b,\vec c\}}\nonumber\\
 &= \min\left\{\norm{\vec b'-\vec b},\norm{\vec b'-\vec c},\norm{\vec a - \vec b},\norm{\vec a-\vec c}\right\},\label{EQ_Values_of_r}
\end{align}
is the distance between the supports of $h_{\vec b,\vec c}(t)$ and $h_{\vec a,\vec b'}(s)$ (see \cref{FIG_doverlap}).
Since each term of $\doverlap$ contributes a sum of an algebraically decaying as $1/r^\alpha$ and an exponential decaying as $e^{-\gamma r}$ terms, it is convenient to evaluate their contributions separately. 

First, let us find the contribution from the algebraically decaying part. 
It is straightforward to find out their contributions to $\doverlap$ when $r$ takes one of the four allowed values.
Depending on which value $r$ takes, we use either Lemma~\ref{LEM_SIMPLE_SUM} or Lemma~\ref{LEM_REP_CON} in \Cref{Sec_math} to evaluate the sums over $\vec b$ and $\vec b'$. 
For example, the contribution from the terms where $r = \norm{\vec b'-\vec b}$ is at most
\begin{align}
&\sum_{\vec a\in A}\sum_{\vec b\neq \vec b'\in B}\sum_{\vec c \in C}  \frac{c (e^{vt}-1)\left(\frac{1}{1-\gamma}\right)^\alpha}{v\norm{\vec b-\vec c}^\alpha\norm{\vec a - \vec b'}^\alpha\norm{\vec b'-\vec b}^\alpha}\nonumber\\
&\leq \frac{c}{v}\lambda_3 \sum_{\vec a\in A}\sum_{\vec b\in B}\sum_{\vec c \in C}  \frac{  (e^{vt}-1)\left(\frac{2}{1-\gamma}\right)^\alpha}{\norm{\vec b-\vec c}^\alpha\norm{\vec a - \vec b}^\alpha}\label{EQ_Lem7_1}\\
&\leq\frac{c}{v}\lambda_3 \lambda_4 \sum_{\vec a\in A}\sum_{\vec c \in C}  \frac{ (e^{vt}-1)\left(\frac{4}{1-\gamma}\right)^\alpha }{\norm{\vec a - \vec c}^\alpha}\\
&\leq \lambda_5(e^{vt}-1)\left(\frac{8}{1-\gamma}\right)^\alpha\frac{ \Phi( A)}{\ell^{\alpha-D-1}}
\end{align}
where $\lambda_3,\lambda_4$ are constants that arise after we use Lemma~\ref{LEM_REP_CON} in \Cref{Sec_math} twice to evaluate the sums over $\vec b'$ and $\vec b$ consecutively, and the sums over $\vec a, \vec c$ have been bounded in the previous section (see \cref{EQ_sumAC_1}).
The constant $\lambda_5$ absorbs both $\lambda_3,\lambda_4$ and the constants from the sums over $\vec a,\vec c$.

On the other hand, if $r = \norm{\vec b'-\vec c}$, we use Lemma~\ref{LEM_REP_CON} to evaluate the sum over $\vec b'$ and Lemma~\ref{LEM_SIMPLE_SUM} for the sum over $\vec b$:
\begin{align}
&\sum_{\vec a\in A}\sum_{\vec b\neq \vec b'\in B}\sum_{\vec c \in C}  \frac{1}{\norm{\vec b-\vec c}^\alpha}\frac{1}{\norm{\vec a - \vec b'}^\alpha}\frac{c (e^{vt}-1)\left(\frac{1}{1-\gamma}\right)^\alpha}{v\norm{\vec b'-\vec c}^\alpha}\nonumber\\
&\leq \frac{c}{v}\lambda_6 \sum_{\vec a\in A}\sum_{\vec b\in B}\sum_{\vec c \in C}  \frac{1}{\norm{\vec b-\vec c}^\alpha}\frac{ (e^{vt}-1)\left(\frac{2}{1-\gamma}\right)^\alpha}{\norm{\vec a - \vec c}^\alpha}\label{EQ_Lem5_2}\\
&\leq\frac{c}{v}\lambda_6\lambda_7 \sum_{\vec a\in A}\sum_{\vec c \in C}  \frac{ (e^{vt}-1)\left(\frac{2}{1-\gamma}\right)^\alpha}{\norm{\vec a - \vec c}^\alpha}\\
&\leq \lambda_8 (e^{vt}-1)\left(\frac{4}{1-\gamma}\right)^\alpha\frac{\Phi( A)}{\ell^{\alpha-D-1}},
\end{align} 
where $\lambda_6,\lambda_7$ come from the uses of \cref{LEM_REP_CON} and \cref{LEM_SIMPLE_SUM} respectively. 
The constant $\lambda_8$ absorbs both $\lambda_6,\lambda_7$ and the constants from the sums over $\vec a,\vec c$.
Repeating for the other values of $r$, we find that the contribution from the algebraically decaying terms in \Cref{EQ_doverlap_two_parts} to $\doverlap$ is at most 
\begin{align}
\lambda_9(e^{vt}-1)\left(\frac{8}{1-\gamma}\right)^\alpha\frac{ \Phi( A)}{\ell^{\alpha-D-1}},\label{EQ_doverlap_alg}
\end{align}
for some constant $\lambda_9$.

Next, let us find the contribution from the exponentially decaying term in \Cref{EQ_doverlap_two_parts}. 
If $r = \norm{\vec b'-\vec b}$, we have
\begin{align}
	&\sum_{\vec a\in A}\sum_{\vec b\neq \vec b'\in B}\sum_{\vec c \in C}  \frac{c(e^{vt}-1)}{v\norm{\vec b-\vec c}^\alpha\norm{\vec a - \vec b'}^\alpha e^{\gamma\norm{\vec b'-\vec b}}}\nonumber\\
	&\leq \frac{c}{v}\lambda_{10}\sum_{\vec a\in A}\sum_{\vec b\in B}\sum_{\vec c \in C}  \frac{(e^{vt}-1)}{\norm{\vec b-\vec c}^\alpha} \Bigg(\frac{ \left(\frac{4}{1-\gamma}\right)^\alpha}{\norm{\vec{a}-\vec{b}}^\alpha}
	+\frac{\norm{\vec a-\vec b}^{D-1}}{e^{\gamma \norm{\vec a-\vec b}}}\Bigg)\nonumber\\
	&\leq \frac{c}{v}\lambda_{10}\lambda_{11}\sum_{\vec a\in A}\sum_{\vec c \in C} (e^{vt}-1)  \Bigg(
	\frac{ \left(\frac{8}{1-\gamma}\right)^\alpha}{\norm{\vec{a}-\vec{c}}^\alpha}
	+\frac{\norm{\vec{a}-\vec{c}}^{2D-2} }{ e^{\gamma\norm{\vec a-\vec c}}}
	\Bigg)\nonumber\\
	&\leq \lambda_{12} (e^{vt}-1) \Phi(A)\Bigg(
	\frac{\left(\frac{16}{1-\gamma}\right)^\alpha}{\ell^{\alpha-D-1}}
	+\frac{\ell^{3D-3}}{e^{\gamma \ell}}\Bigg),
	\label{EQ_doverlap_exp}
\end{align}
where we have applied \Cref{LEM_REP_CON_MIX} in \Cref{Sec_math} to obtain the first inequality, \Cref{LEM_REP_CON_MIX} twice again and \Cref{LEM_REP_CON} to get the second inequality, and then \Cref{LEM_SIMPLE_SUM} and \Cref{LEM_SIMPLE_SUM_exp} for the sums over $\vec a,\vec c$ similarly to \Cref{SM_Subsec_dtrunc}.
The constants $\lambda_{10},\lambda_{11}$ arise from the applications of the lemmas and are absorbed into a constant $\lambda_{12}$.
We note that the constant $\gamma$ in the last three lines are different from the one in the first line (see \Cref{LEM_REP_CON_MIX} for details). 
However, they both are constants that can be chosen arbitrarily between 0 and 1. Therefore, we denote them by the same constant $\gamma$ for convenience.

Repeating the argument for other choices of $r$ in \cref{EQ_Values_of_r}, we find that the contribution from the exponentially decaying terms to $\doverlap$ is still at most the right hand side of \cref{EQ_doverlap_exp}.

Combining \cref{EQ_doverlap_alg} and \cref{EQ_doverlap_exp}, we have
\begin{align}
	\norm{\doverlap} \leq \lambda_{13} (e^{vt}-1) \Phi(A)\Bigg(
	\frac{\left(\frac{16}{1-\gamma}\right)^\alpha}{\ell^{\alpha-D-1}}
	+\frac{\ell^{3D-3}}{e^{\gamma \ell}}\Bigg),
\end{align}
for a constant $\lambda_{13}$.
Since $\ell^{D-1}\leq \frac{(D-1)!}{\epsilon^{D-1}}e^{\epsilon \ell}$ for any arbitrary small positive constant $\epsilon$, we can upper bound
\begin{align}
	&\norm{\doverlap}\leq c_{\ov} (e^{vt}-1) \Phi(A)\Bigg(
	\frac{\left(\frac{16}{1-\gamma}\right)^\alpha}{\ell^{\alpha-D-1}}
	+\frac{1}{e^{\gamma \ell}}\Bigg), \label{EQ_cov_def}
\end{align}
where we have absorbed $\epsilon$ into the definition of $\gamma$ and $c_{\ov}$. 
This completes the estimation of $\doverlap$.
\Section{Error propagation from generating function}
\label{APP_ERR_PROP}
In this section, we reproduce a lemma in Ref.~\cite{Haah} which shows how the error in approximating the generating function $\mathcal G_W$ propagates to an error of the unitary $W_t$ in \cref{EQ_G_def}.
Suppose we approximate $\mathcal G_W$ by $\mathcal G'_W$ such that
\begin{align}
	\norm{\mathcal G_W-\mathcal G'_W}\leq f(t)\delta,
\end{align}
for some function of time $f(t)$ and $\delta$ is time-independent.
We shall prove that the unitary $W'_t$ generated by $\mathcal G'_{W}$ approximates $W_t$ with error
\begin{align}
	\norm{W'_t-W_t}\leq \delta \int_0^t \d s f(s).
\end{align}
\begin{proof}
By simple differentiation, we have
\begin{align}
	\norm{W^\dag_t W'_t-\mathbb I}&=\Norm{\int_0^t \d s \frac{\d}{\d s} (W^\dag_s W'_s)}\\
	&=\Norm{\int_0^t \d s W^\dag_s(G_{W}-G'_{W})W'_s}\\
	&\leq\int_0^t \d s \norm{W^\dag_s(G_{W}-G'_{W})W'_s}\\
	&=\int_0^t \d s \norm{G_{W}-G'_{W}}\\
	&\leq \delta \int_0^t \d s f(s).
\end{align}
\end{proof}
\section{Proof of the Lieb-Robinson bound for long-range interactions}
\label{APP_LR_PROOF}
We present a more detailed proof of \Cref{TH_LR} in this section.
The key ingredient in the proof of Theorem~\ref{TH_LR} is the following lemma.
\begin{lemma}
	\label{LEM_LR_exist}
	Denote by $\B_{r}=\left\{\veci\in \Lambda: \Norm{\veci}\leq r\right\}$ a $D$-ball of radius $r$ centered around the origin.
	Let $O_X$ be an observable supported on $X = \B_{r_0}$ with $r_0$ being finite.
	For each $\U{\Lambda}{0,T}$ and a positive integer $M$, there exists a unitary $\tilde U$ supported on a $D$-ball $\B_r$ with $r=r_0+M\ell$ such that 
	\begin{align}
	&\Norm{\Udag{\Lambda}{0,T}O_X\U{\Lambda}{0,T}-\tilde U^\dag O_X \tilde U}
	\leq b_1 M e^{vt}\left(r-\ell\right)^{D-1}\xi_{\alpha}(\ell),
	\end{align}
	where $b_1$ is a constant, $t = T/M$ and $\ell\in(0,R)$ is a free parameter.
\end{lemma}
\begin{proof}
	\begin{figure}[t]
		\includegraphics[width=0.35\textwidth]{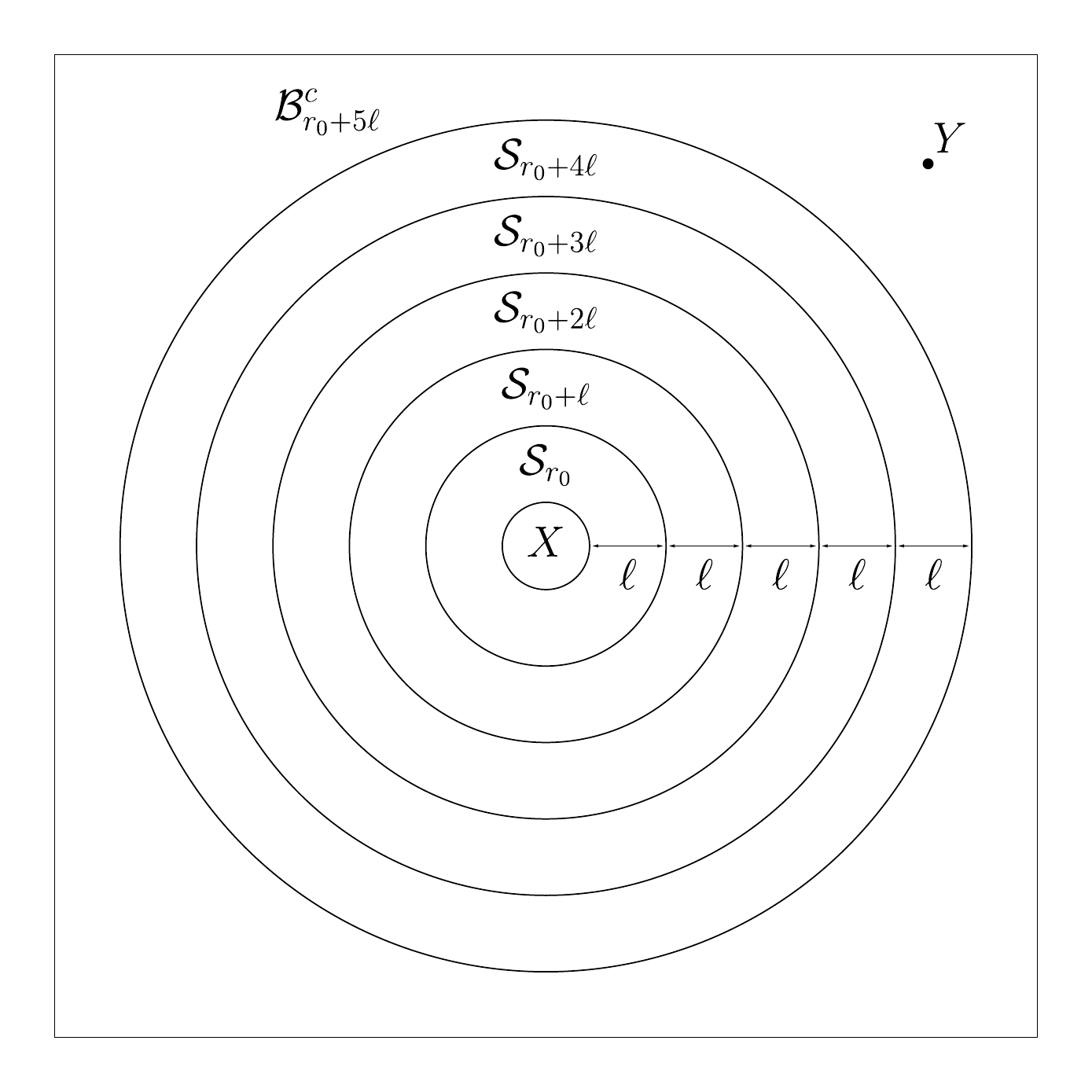}
		\caption{An example of the subset $X=\B_{r_0}$ and five shells $\mathcal S_r$ for $r = r_0,r_0+\ell,\dots,r_0+4\ell$. The operator $O_Y$ is supported on $Y$, which lies on $\B_{r_0+5\ell}^c$, the complement of the ball $\B_{r_0+5\ell}$.}
		\label{FIG_TH-LR-shell}
	\end{figure}
	We shall prove the lemma by constructing the unitary $\tilde U$.
	In addition to $\mathcal{B}_r$ above, we define
	\begin{align}
	\mathcal S_{r} = \mathcal{B}_{r+\ell}\setminus \mathcal{B}_r
	\end{align}
	to be a shell consisting of sites between $r$ and $r+l$ away from the origin (\cref{FIG_TH-LR-shell}).

	We divide $[0,T]$ into $M$ equal time intervals, namely $[(M-k-1)t,(M-k)t]$ for $k = 0,\dots,M-1$, where $t = T/M$. 
	The unitary $\U{\Lambda}{0,T}$ then naturally decomposes into a product of unitaries $\U{\Lambda}{k}\equiv\U{\Lambda}{(M-k-1)t,(M-k)t}$:
	\begin{align}
	\U{\Lambda}{0,T}= \U{\Lambda}{0}\U{\Lambda}{1}\dots\U{\Lambda}{M-1}.
	\end{align}
	We now use Lemma~\mainref{LEM_BREAK_D} to further decompose each $\U{\Lambda}{k}$ into evolutions of subsystems.
	We start with $k = 0$ and use Lemma~\ref{LEM_BREAK_D} with $A\rightarrow X=\B_{r_0},B\rightarrow \mathcal S_{r_0}$, and $C\rightarrow \mathcal B_{r_0+\ell}^c$ (Fig.~\ref{FIG_TH-LR-demo}) to decompose $\Udag{\Lambda}{0}$ (instead of $\U{\Lambda}{0}$):
	\begin{align}
	&\Norm{\Udag{\Lambda}{0}-\Udag{\B_{r_0+\ell}}{0}\U{\mathcal S_{r_0}}{0}\Udag{\mathcal{B}_{r_0}^c}{0}}\nonumber\\
	&\leq c_0e^{vt} \Phi(\B_{r_0})\xi_{\alpha}(\ell),
	\end{align}
	where again $\mathcal{B}_{r_0}^c$ denotes the complement subset $\mathbb R^D\setminus \mathcal{B}_{r_0}$, and $\Phi(\mathcal{B}_{r_0})$ is the boundary area of $\mathcal{B}_{r_0}$.
	This choice of decomposition allows us to eliminate the contribution to the evolution from the terms of the Hamiltonian that commute with $X$, i.e.\ those supported entirely on $\B_{r_0}^c$.
	Explicitly, we have:
	\begin{align}
	&\Udag{\Lambda}{0} O_X \U{\Lambda}{0}\nonumber\\
	&\approx\Udag{\B_{r_0+\ell}}{0} \U{\mathcal S_{r_0}}{0} \Udag{\mathcal{B}_{r_0}^c}{0}O_X \U{\mathcal{B}_{r_0}^c}{0}\Udag{\mathcal S_{r_0}}{0}\U{\B_{r_0+\ell}}{0}\\
	&=\Udag{\B_{r_0+\ell}}{0} O_X \U{\B_{r_0+\ell}}{0}=\tilde U_0^\dag O_X \tilde U_0,
	\end{align}
	where $\tilde U_0\equiv \U{\B_{r_0+\ell}}{0}$ is supported entirely on $\B_{r_0+\ell}$.

	If we repeat the above argument for $\U{\Lambda}{1}$ but with $O_X$ replaced by $\tilde U_0^\dag O_X \tilde U_0$, we can approximate
	\begin{align}
	\Udag{\Lambda}{1} \tilde U_0^\dag O_X \tilde U_0 \U{\Lambda}{1} \approx \tilde U_1^\dag \tilde U_0^\dag O_X \tilde U_0 \tilde U_1,
	\end{align}
	for some $\tilde U_1$ supported entirely on $\B_{r_0+2\ell}$.
	The error of this approximation is at most 
	$
	c_0 e^{vt} \Phi(\B_{r_0+\ell})\xi_{\alpha}(\ell).
	$	
	
	By induction to all $k=2,\dots,M-1$, we can construct $\tilde U = \tilde U_0 \tilde U_1\dots \tilde U_{T-1}$ such that
	\begin{align}
	\Udag{\Lambda}{0,T}O_X \U{\Lambda}{0,T}\approx \tilde U^\dag O_X \tilde U,
	\end{align}
	where the overall error is at most
	\begin{align}
	&\sum_{k=0}^{M-1}c_0 e^{vt} \Phi(\B_{r_0+k\ell})\xi_{\alpha}(\ell)\nonumber\\
	&\leq Mc_0 e^{vt} \Phi(\B_{r_0+(M-1)\ell})\xi_{\alpha}(\ell)\\
	&\leq \underbrace{c_0 \frac{2\pi^{\frac{D}{2}}}{\Gamma(\frac{D}{2})}}_{\equiv b_1} M e^{vt}\left(r_0+\left(M-1\right)\ell\right)^{D-1}\xi_{\alpha}(\ell),
	\end{align}
	and where we have replaced the surface area $\Phi(\B_r)$ of a $D$-ball $\B_r$ by $\frac{2\pi^{\frac{D}{2}}}{\Gamma(\frac{D}{2})} r^{D-1}$ and $M$ by $T/t$.
	Also by induction, the unitary $\tilde U$ is supported entirely on $\B_{r_0+M\ell}$. 
	Therefore the lemma follows.
\end{proof}

We are now ready to prove our Lieb-Robinson bound in Theorem~\ref{TH_LR}.
Without loss of generality, we assume the origin is in $X$. 
Since $\Norm{X}=\O{1}$, there exists $r_0=\O{1}$ such that $X$ is a subset of $\B_{r_0}$. 
By \Cref{LEM_LR_exist}, there exists a unitary $\tilde U$ supported entirely on a $D$-ball $\mathcal{B}_{r}$ with $r = r_0+M\ell$ such that 
\begin{align}
&\epsilon = \Norm{\Udag{\Lambda}{0,T}O_X\U{\Lambda}{0,T}-\tilde U^\dag O_X \tilde U}\nonumber\\
&\leq b_1 M e^{vt}\left(r-\ell\right)^{D-1}\xi_{\alpha}(\ell).\label{EQ_LR_app_err}
\end{align}
If we choose the number of time slices ($M$) and the block size ($\ell$) such that $r\leq R+r_0$, the set $Y$ will lie outside the support $\B_r$ of $\tilde U^\dag O_X \tilde U$, and therefore $\tilde U^\dag O_X \tilde U$ will commute with $O_Y$.
Note that for a fixed value of $M$, the error should decrease with a larger value of $\ell$.
Therefore, to prove the strongest bound, we should choose $\ell$ as large as possible, i.e., $\ell = R/M$, and hence $M = R/\ell$.  
Substituting the value of $M$ and $t=T/M$ into \cref{EQ_LR_app_err}, we obtain the error bound in terms of $\ell$ alone:
\begin{align}
\epsilon &\leq b_2 \frac{R}{\ell} (e^{\frac{vT\ell}{R}}-1)\left(1+R-\ell\right)^{D-1}\xi_{\alpha}(\ell), \label{EQ_LR_err_l}
\end{align}
where $b_2 = b_1 r_0^{D-1}$ is a finite constant.
We note that the above bound is valid for all values of $\ell \leq R$.
The tightest bound can therefore be obtained by choosing a value for $\ell$ that minimizes the above expression.
Our intuition and numerical evidence suggest that this happens when $\ell\sim \frac{R\alpha}{vT}$, so in the below analysis, we aim to choose $\ell$ as close to this value as possible. 

To proceed, we consider two regimes of time $T$, when $vT\geq \alpha$ and when $vT<\alpha$.
In the former regime, we choose $\ell = \frac{ R\alpha}{vT}\leq R$ and substitute into \cref{EQ_LR_err_l} to get
\begin{align}
\epsilon &\leq  b_2\frac{vT}{\alpha} (e^{\alpha}-1) \left(1+R\left(1-\frac{\alpha}{vT}\right)\right)^{D-1}\xi_{\alpha}\left(\frac{R\alpha}{vT}\right)\nonumber\\
&\leq  \underbrace{b_2v2^{D-1}}_{\equiv c_{\lr}}T  R^{D-1}\xi_{\alpha}\left(\frac{R\alpha}{vT}\right),\label{EQ_epsilon_bound}
\end{align}
where we have used $\frac{e^\alpha-1}{\alpha}\leq 1$ for all $\alpha\geq 1$, $1-\frac{\alpha}{vT}\leq 1$ and $1+R\leq 2R$.
In particular, if $\frac{R\alpha}{vT}>x_0$, where $x_0$ is the larger solution of $x^{\alpha-D-1}=e^{\gamma x}$, the algebraically decaying term in $\xi_\alpha$ dominates the exponentially decaying one, and therefore
\begin{align}
\xi_{\alpha}\left(\frac{R\alpha}{vT}\right) 
&= 
\left(\frac{16}{1-\gamma}\right)^\alpha\frac{1}{\left(\frac{R\alpha}{vT}\right)^{\alpha-D-1}}	
+e^{-\gamma \left(\frac{R\alpha}{vT}\right)},\\
&\leq 2
\left(\frac{16}{1-\gamma}\right)^\alpha\frac{1}{\left(\frac{R\alpha}{vT}\right)^{\alpha-D-1}}\\
&= 2
\left(\frac{16}{1-\gamma}\right)^\alpha\left(\frac{v}{\alpha}\right)^{\alpha-D-1}\left(\frac{T}{R}\right)^{\alpha-D-1}.\label{EQ_xia_gtr}
\end{align}
Combining \cref{EQ_epsilon_bound,EQ_xia_gtr}, we obtain a bound on the commutator norm:
\begin{align}
\mathcal{C}(T,R)\leq \epsilon \leq  c_{\lr,\alpha}
\frac{T^{\alpha-D}}{R^{\alpha-2D}},
\end{align}
where
\begin{align}
	c_{\lr,\alpha}\coloneqq 2c_{\lr} 
	\left(\frac{16}{1-\gamma}\right)^\alpha\left(\frac{v}{\alpha}\right)^{\alpha-D-1}.\label{EQ_clra}
\end{align}
The light cone implied by the bound is
\begin{align}
	T \gtrsim R^{\frac{\alpha-2D}{\alpha-D}}.
\end{align}
In the limit $\alpha\to\infty$, the exponent of the light cone converges to one at a rate given by
\begin{align}
	\mu = \lim_{\alpha\to\infty} \frac{\abs{\frac{\alpha+1-2D}{\alpha+1-D}-1}}{\abs{\frac{\alpha-2D}{\alpha-D}-1}} = 1.\label{EQ_converge_rate}
\end{align}

On the other hand, if $vT<\alpha$, we simply choose $\ell = R$.
\Cref{EQ_LR_err_l} then becomes
\begin{align}
	\mathcal{C}(T,R)\leq \epsilon\leq b_2(e^{vT}-1)\xi_\alpha(R).
\end{align} 
Therefore, we arrive at the Lieb-Robinson bound in \Cref{TH_LR} with $\tilde c_\text{lr}=b_2$.
\section{Proof of Lemma~\mainref{LEM_HHKL_existence} in higher dimensions}
\label{SM_Sec_HHKL}
In this section, we discuss the construction of the circuit in Lemma~\ref{LEM_HHKL_existence} that generalizes the lemma to higher dimensions.
Similar to the $D=1$ case, we first break the unitary into $\O{T}$ unitaries $\exp(-iHt)$ for some $t=\O{1}$.
We then use an algorithm consisting of $D$ steps to break the simulation of $\exp(-iHt)$ into simulations of Hamiltonians on smaller hypercubes of size at most $2\ell$.
In the first of the $D$ steps, we cut the $D$-dimensional lattice into $L/\ell$ layers, each with the same thickness $\ell$, a parameter to be chosen later.
In this step, the cross section of the cut is $L^{D-1}$.
Therefore, by \cref{LEM_BREAK_D}, each time a new layer is generated, we accumulate an error of 
\[\O{\frac{L^{D-1}}{\ell^{\alpha-D-1}}}.\]
For $TL/\ell$ layers of the first step, the accumulated error will be $\epsilon^{(1)}=\O{{TL^{D}}/{\ell^{\alpha-D}}}$.\\

Next, for each of the $\O{{TL}/{\ell}}$ layers of $D-1$ dimensions, we break them again into $L/\ell$ layers of $D-2$ dimensions. Using Lemma~\ref{LEM_BREAK_D} with a cross section $L^{D-2}$, we find the error of the second step 
\begin{align}
\epsilon^{(2)}=\frac{TL}{\ell}\frac{L}{\ell}\O{\frac{L^{D-2}}{\ell^{\alpha-(D-1)-1}}}=\O{\frac{TL^D}{\ell^{\alpha-D+2}}},
\end{align}
which decreases with $\ell$ faster than the error of the first step.

More explicitly, in the $k$th of the $D$ steps, the error is $\epsilon^{(k)} = \O{{L^{D}}/{\ell^{\alpha-D-2k}}}$, which is dominated by the error in the first step for all $k>1$.
Therefore, the error of cutting the $D$-dimensional lattice of size $L$ into $L^D/\ell^D$ subsystems is still $\O{{TL^{D}}/{\ell^{\alpha-D}}}$.
To meet a fixed total error $\epsilon$, we need to choose $\ell\propto \left(TL^D/\epsilon\right)^{\frac{1}{\alpha-D}}$.
The geometrical constraint $\ell<L$ requires $\alpha>2D$.
Finally, simulating each of the $\O{{TL^D}/{\ell^D}}$ subsystems using the LCU algorithm up to ${\epsilon \ell^D}/{(TL^D)}$ accuracy requires $\O{\ell^{3D}\log\left({TL^D}/{\epsilon \ell^D}}\right)$ quantum gates.
Therefore, the overall gate complexity of the algorithm is 
\begin{align}
G_D = \O{\frac{(Tn)^{1+\frac{2D}{\alpha-D}}}{\epsilon^{\frac{2D}{\alpha-D}}}\log\frac{Tn}{\epsilon}}.
\end{align}

\section{Estimation of the actual gate count}
\label{APP_gate_count}
In this section, we describe how we estimate the actual gate count of the HHKL algorithm and the QSP algorithm in simulating one-dimensional power-law systems. 

The direct implementation of the QSP algorithm requires computing a sequence of rotation angles on a classical computer, which is prohibitive for large-size Hamiltonian simulation. Instead, we use a suboptimal approach described in Ref.~\cite{ChildsMNRS2017}. To simulate $H=\sum_{j=1}^L\beta_jH_j$ for time $t$ and accuracy $\epsilon$, where $L$ is the number of terms in the Hamiltonian, $\beta_j\geq 0$ and $H_j$ are both unitary and Hermitian, we divide the entire evolution into $r$ segments. We choose $r$ sufficiently large so that each segment is short enough for the classical preprocessing. Specifically, we choose
\begin{equation}
r=\bigg\lceil\frac{\sum_j\beta_j t}{\tau_{\max}}\bigg\rceil
\end{equation}
and $\tau_{\max}=1000$ \cite{Haah2018}. Within each segment, we choose $q$ to be the smallest positive integer satisfying
\begin{equation}
\frac{4(\sum_j\beta_j t/r)^q}{2^q q!}\leq\frac{\epsilon}{8r},
\end{equation}
so that the overall error is at most $\epsilon$. This gives $M=2(q-1)$ phased iterates within each segment~\cite{ChildsMNRS2017}.

The number of elementary operations of each phased iterate is $\log(L)+4L+8L$. Here, the first term corresponds to the reflection along an $L$-dimensional state $|0\rangle$; the second term costs the preparation/unpreparation of an $L$-dimensional state; and the third term is the cost of selecting $L$ two-body operators. We thus estimate the gate complexity of the QSP algorithm as $\big(\log(L)+12L\big)rM$.

Next, in order to determine the gate count of the HHKL algorithm, we need an estimate for the error of the unitary decomposition in \cref{LEM_BREAK_D}.
Recall that for $D=1$, the error given by our analysis is $b/\ell^{\alpha-2}$, where $b$ is a constant that can be estimated numerically by computing the actual error for small values of $\ell$ and extrapolating for larger $\ell$.

\begin{figure}[t]
	\includegraphics[width=0.45\textwidth]{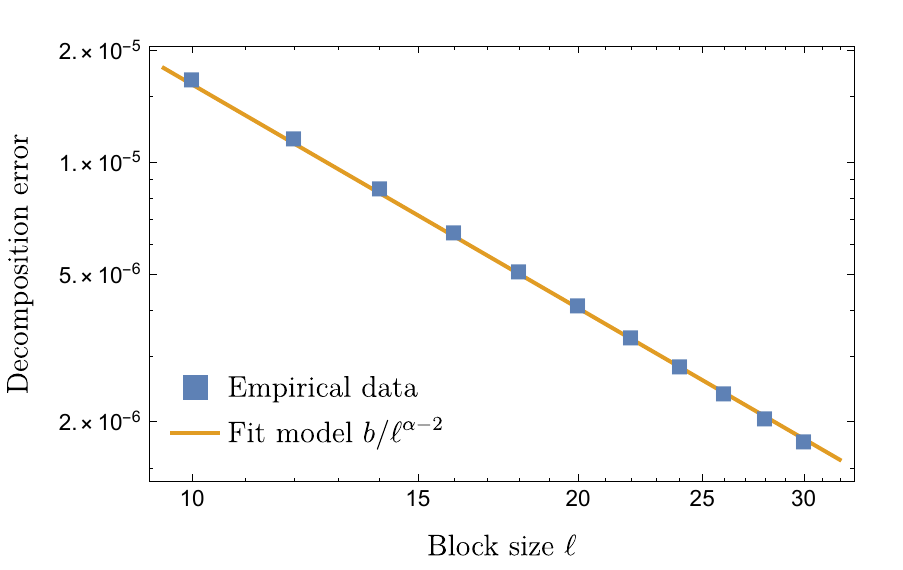}
	\caption{ 
		The empirical error of the unitary decomposition in \cref{LEM_BREAK_D}, computed for the single-excitation one-dimensional Heisenberg chain ($\alpha = 4$) in \cref{EQ_Heisenberg} at different values of $\ell$.
		The system size is fixed at $n=300$ and the evolution time at $t=0.01$.
		We fit the data (blue square) to the theoretical model $b/\ell^{\alpha-2}$ and obtain $b = 1.62\times 10^{-3}$.
	}
	\label{FIG_ell_scaling}
\end{figure}

Since simulating the evolution of a generic system is classically intractable even for a moderate system size, we study only the one-dimensional Heisenberg model given in \cref{EQ_Heisenberg} and restrict our calculation to the single-excitation subspace.
In \cref{FIG_ell_scaling}, we plot the error of the unitary decomposition in \cref{LEM_BREAK_D} at several different values of $\ell$ (for system size $n=300$ and evolution time $t=0.01$).
The scaling of the error agrees well with our prediction.
By fitting the data to $b/\ell^{\alpha-2}$, we obtain an estimate $b = 1.62 \times 10^{-3}$. 

Recall that there are $T/t$ time slices in the HHKL algorithm.
In each time slice, there are $n/\ell$ blocks of size $\ell$ and $2n/(2\ell)$ blocks of size $2\ell$.
To meet the total error at most $\epsilon$, we need to choose (see also \cref{EQ_error_scaling_ell})
\begin{align}
	\ell = \left(\frac{T}{t}\frac{2nb}{\epsilon}\right)^{\frac{1}{\alpha-1}}.
\end{align}
By multiplying the number of blocks by the gate count for using QSP to simulate a single block, we arrive at the total gate count presented in \cref{FIG_actual-gate-count}.

\section{Numerical performance of the product formula}
\label{Sec_num}
This section includes the numerical performance of the fourth-order product formula (PF4) used to simulate the evolution of the system given in \cref{EQ_Heisenberg} for time $T=n$.
We plot this numerical performance as well as the theoretical estimates for the gate counts of the PF4, LCU, QSP, and HHKL algorithms in \cref{FIG_COMPARE}.
\vspace{0.01in}
\begin{figure}[H]
	\includegraphics[width=0.45\textwidth]{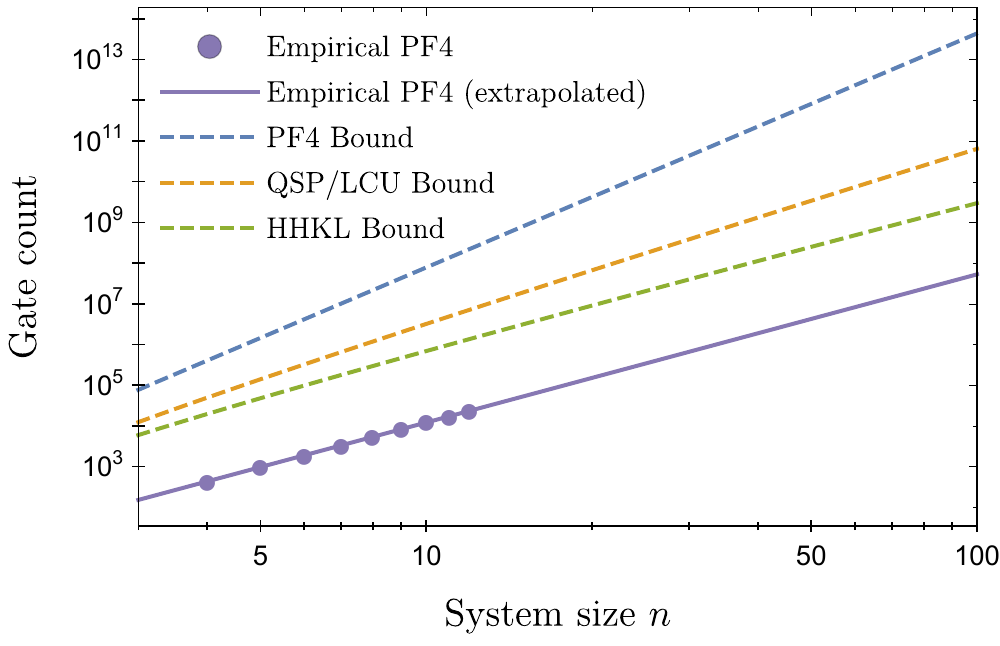}
	\caption{ 
		The empirical gate count of PF4 (purple dots) from $n = 4$ to $n = 12$, extrapolated to larger system sizes (solid, purple), for simulating the dynamics of the Hamiltonian in Eq.~\eqref{EQ_Heisenberg} for time $T=n$ at a fixed error tolerance.
		The error bars are smaller than the size of the markers and hence not visible in the plot.
		Also shown in dashed lines are the \emph{slopes} of the gate counts of several advanced algorithms for comparison.
		These slopes represent the scaling of the gate counts as functions of $n$.	
		Their $y$-intercepts, which represent a constant multiplicative factor, should be ignored.
	}
	\label{FIG_COMPARE}
\end{figure}

\textbf{}
\section{Mathematical tools}
\label{Sec_math}
This section contains a collection of mathematical results omitted from the previous sections.
In \Cref{SM_Subsec_std_sum}, we present the upper bounds on standard sums we use in the proof of \cref{LEM_BREAK_D} in \cref{APP_LEM_1_PROOF}.
In \Cref{Subsec_para}, we show how we estimate the sum over the convex set $A$ in \cref{EQ_sumveca} by parameterizing the elements of the set by their distance to the boundary of $A$. 
We also note that we use the same notation ``$\lambda$'' for constants that appear in different lemmas.

\subsection{Standard sums}\label{SM_Subsec_std_sum}
In this section, we present upper bounds on a a few standard sums used in the previous sections.
Specifically, we use \Cref{LEM_SIMPLE_SUM} to bound \cref{EQ_sumAC_1}, \cref{EQ_Lem5_2},
\Cref{LEM_SIMPLE_SUM_exp} to bound \cref{EQ_Lem7_1}, \cref{EQ_doverlap_exp},
\Cref{LEM_REP_CON} to bound \cref{EQ_Lem5_2}, \cref{EQ_doverlap_exp}, 
and \Cref{LEM_REP_CON_MIX} to bound \cref{EQ_doverlap_exp}.

\begin{lemma}
	\label{LEM_SIMPLE_SUM}
	Let $\Lambda$ be a $D$-dimensional lattice and $\vec r$ be the coordinates of sites in $\Lambda$.
	For $\alpha>D+1$ and $R> \sqrt D$, there exists a constant $\lambda$ that may depend on $D$ but not on $R,\alpha$ such that:
	\begin{align}
	\sum_{\substack{\vec r\in \Lambda\\\Norm{\vec r}\geq R}}\frac{1}{\Norm{\vec r}^\alpha} \leq \frac{\lambda}{(R-\sqrt D)^{\alpha-D}}.
	\end{align}
	In particular, it implies that the sum $\sum_{\substack{\vec r\in \Lambda}}$ converges for all $\alpha>D$.
\end{lemma}
\begin{proof}
	The proof of this bound is straightforward. 
	For simplicity, we first assume none of the coordinates of $\vec r$ is zero. 
	Since $\frac{1}{x^\alpha}$ is a decreasing function of $x$ for all $\alpha>0$, we can always bound the sum over such $\vec r$ by an integral
	\begin{align}
	&{\sum_{\substack{\vec r\in \Lambda \\ \Norm{\vec r}\geq R}}}' \frac{1}{\Norm{\vec r}^\alpha} \leq \int_{\Norm{\vec r}\geq R-\sqrt{D}} \frac{\d^D \vec r}{\Norm{\vec r}^\alpha}\nonumber\\
	&=\frac{2\pi^{\frac D 2}}{\Gamma(\frac D2)} \int_{R-\sqrt{D}}^\infty \frac{\d r}{r^{\alpha-D+1}}
	\leq \frac{g(D)}{(R-\sqrt D)^{\alpha-D}},
	\end{align}
	where $\sum'$ denotes the sum over $\vec r$ with no zero coordinate and $g(D)\equiv 2\pi^{\frac D 2}\big/\Gamma(\frac{D}{2})$.
	
	Next, consider $\vec r$ with exactly one zero coordinate. 
	These sites lie on $D$ hyperplanes, each of dimension $(D-1)$. 	Therefore the contribution from them can be evaluated using the above integral with $D\rightarrow D-1$: 
	\begin{align}
	\frac{Dg(D-1)}{(R-\sqrt{D-1})^{\alpha-D+1}}< \frac{Dg(D-1)}{(R-\sqrt{D})^{\alpha-D}}.
	\end{align}
	By repeating this argument for the sums over $\vec r$ with different number of zero coordinates, we arrive at 
	\begin{align}
	{\sum_{\substack{\vec r\in \Lambda \\ \Norm{\vec r}\geq R}}} \frac{1}{\Norm{\vec r}^\alpha} \leq \frac{\lambda}{(R-\sqrt D)^{\alpha-D}},
	\end{align}
	where $\lambda = \sum_{d = 0}^{D} \binom{D}{d}g(D-d)$ is a constant independent of $R$.
\end{proof}
\begin{lemma}
	\label{LEM_SIMPLE_SUM_exp}
	Let $\Lambda$ be a $D$-dimensional lattice and $\vec r$ be the coordinates of sites in $\Lambda$.
	For all $R>0$, there exists a constant $\lambda$ that may depend on $\beta,D$ but not on $R$ such that:
	\begin{align}
	\sum_{\substack{\vec r\in \Lambda\\\Norm{\vec r}\geq R}}\frac{\norm{\vec r}^\beta}{e^{\Norm{\vec r}}} \leq \frac{\lambda R^{\beta+D-1}}{e^{R}},
	\end{align}
	where $\beta$ is a positive constant.
	In particular, it also implies that the sum $\sum_{\substack{\vec r\in \Lambda}}$ converges.
\end{lemma}
\begin{proof}
	The proof of this lemma follows the same idea as of \Cref{LEM_SIMPLE_SUM}.
	However, note that the function $x^\beta e^{-x}$ is a decreasing function of $x$ only when $x\geq x_0$ for some $x_0$ that depends only on $\beta$.
	Therefore, if $R\geq x_0$, we follow the exact same lines as in the proof of \Cref{LEM_SIMPLE_SUM}.
	For example, if none of the coordinates of $\vec r$ is zero, we can bound
	\begin{align}
	&{\sum_{\substack{\vec r\in \Lambda \\ \Norm{\vec r}\geq R}}}' \frac{\Norm{\vec r}^\beta}{e^{\Norm{\vec r}}} \leq \int_{\Norm{\vec r}\geq R-\sqrt{D}}  \frac{\Norm{\vec r}^\beta}{e^{\Norm{\vec r}}}\d^D\vec r\nonumber\\
	&=\frac{2\pi^{\frac D 2}}{\Gamma(\frac{D}{2})} \int_{R-\sqrt{D}}^\infty \frac{r^{\beta+D-1}\d r}{e^r}\\
	&\leq \lambda_1\frac{(R-\sqrt D)^{\beta+D-1}}{e^{R-\sqrt D}}
	\leq \lambda_2\frac{R^{\beta+D-1}}{e^{R}},
	\end{align}
	for some constants $\lambda_1,\lambda_2$ that depend only on $\beta,D$.
	
	On the other hand, if $R<x_0$, we consider
	\begin{align}
	\lambda = \max\left\{\frac{e^R}{R^{\beta+D-1}},\sum_{\substack{\vec r\in \Lambda\\\Norm{\vec r}\geq R}}\frac{\norm{\vec r}^\beta}{e^{\Norm{\vec r}}} \right\}.
	\end{align}
	The lemma should follow if we can argue that $\lambda$ can be chosen independently of $R$.
	Indeed, since $1\leq R<x_0$ and from the previous calculation, we know that the sum over $\vec r$ converges to a constant that depends only on $\beta,D$. 
	This concludes the proof of \Cref{LEM_SIMPLE_SUM_exp}.
\end{proof}
\begin{lemma}
	\label{LEM_REP_CON}
	Let $\vec a,\vec b,\vec c$ be three distinct sites in a $D$-dimensional lattice $\Lambda$. 
	For all $\alpha>D$, 
	\begin{align}
	\sum_{\vec b\in \Lambda} \frac{1}{\norm{\vec a-\vec b}^\alpha}\frac{1}{\norm{\vec b-\vec c}^\alpha}\leq \frac{\lambda 2^\alpha}{\norm{\vec{a}-\vec{c}}^\alpha},
	\end{align}
	where $\lambda$ is a constant independent of $\vec a, \vec c,\alpha$.
\end{lemma}
\begin{proof}
	A proof of the lemma is presented in Ref.~\cite{HK}.
\end{proof}
\begin{lemma}
	\label{LEM_REP_CON_MIX}
	Let $\vec a,\vec b,\vec c$ be three distinct sites in a $D$-dimensional lattice $\Lambda$. 
	For all $\alpha>D$, $\gamma\in(0,1)$, and positive integers $\beta\in \mathbb{Z}^+$, there exists a constant $\gamma'\in(0,1)$ such that
	\begin{align}
	&\sum_{\vec b\in \Lambda} \frac{1}{\norm{\vec a-\vec b}^\alpha}\frac{\norm{\vec b-\vec c}^\beta}{e^{\gamma\norm{\vec b-\vec c}}}
	\leq \frac{\lambda \left(\frac{4}{1-\gamma'}\right)^\alpha}{\norm{\vec{a}-\vec{c}}^\alpha}
	+\frac{\lambda'\norm{\vec a-\vec c}^{\beta+D-1}}{ e^{\gamma'\norm{\vec a-\vec c}}},
	\end{align}
	where $\lambda,\lambda'$ are constants that may depend on $\beta,D$, but not on $\vec a, \vec c,\alpha$.
\end{lemma}
\begin{proof}
	Without loss of generality, assume $\vec c = 0$.
	Let $\ell = \norm{\vec c - \vec a}=\norm{\vec a}$ be the distance between $\vec c$ and $\vec a$.
	We need to prove
	\begin{align}
	\sum_{\vec b\in \Lambda} \frac{1}{\norm{\vec a-\vec b}^\alpha}\frac{\norm{\vec b}^\beta}{e^{\gamma\norm{\vec b}}}
	\leq \frac{\lambda \left(\frac{4}{1-\gamma'}\right)^\alpha}{\ell^\alpha}
	+\frac{\lambda'\ell^{\beta+D-1}}{e^{\gamma' \ell}}.
	\end{align}
	
	Let $\B_{\mu \ell}$ be a $D$-ball of radius $\mu \ell$ centered around $\vec c$ for some arbitrary constant $\mu\in(0,1)$.
	We shall divide the sum over $\vec b$ into two regimes, corresponding to 
	$\vec b$ inside and outside $\B_{\mu \ell}$.
	
	In the first regime where $\vec b$ is inside $\B_{\mu \ell}$, we can show using the triangle inequality that $\norm{\vec a-\vec b}\geq (1-\mu) \ell$. Therefore, the sum over these $\vec b$ can be bounded by
	\begin{align}
	\frac{1}{\left((1-\mu)\ell\right)^\alpha}\sum_{\vec b\in \B_{\mu \ell}}\frac{\norm{\vec b}^\beta}{e^{\norm{\vec b}}}\leq \frac{\lambda \left(\frac{2}{1-\mu}\right)^\alpha}{\ell^\alpha},\label{EQ_REP_CON_MIX_1}
	\end{align} 
	where we have used the fact that $\sum_{\vec b\in \B_{\mu \ell}}\frac{\norm{\vec b}^\beta}{e^{\norm{\vec b}}}$ converges and is bounded by a constant $\lambda$ which may depend only on $D,\beta$.
	
	In the second regime, we bound $\norm{\vec a-\vec b}\geq 1$ to obtain
	\begin{align}
	\sum_{\vec b\notin \B_{\mu \ell}} \frac{1}{\norm{\vec a-\vec b}^\alpha}\frac{\norm{\vec b}^\beta}{e^{\gamma\norm{\vec b}}}
	&\leq \sum_{\substack{\vec b\notin \B_{\mu \ell}}} \frac{\norm{\vec b}^\beta}{e^{\gamma\norm{\vec b}}}\nonumber\\
	&\leq \lambda' \frac{ \ell^{\beta+D-1}}{2^\alpha e^{\gamma \mu \ell}},\label{EQ_REP_CON_MIX_2}
	\end{align}
	where the last sum is bounded using \Cref{LEM_SIMPLE_SUM_exp} and noting that $\mu<1$.
	
	Combining \cref{EQ_REP_CON_MIX_1}, \cref{EQ_REP_CON_MIX_2}, we arrive at a bound
	\begin{align}
		\sum_{\vec b\in \Lambda} \frac{1}{\norm{\vec a-\vec b}^\alpha}\frac{\norm{\vec b}^\beta}{e^{\gamma\norm{\vec b}}}
		\leq \frac{\lambda \left(\frac{2}{1-\mu}\right)^\alpha}{\ell^\alpha}
		+\frac{\lambda'\ell^{\beta+D-1}}{e^{\gamma \mu \ell}}.
	\end{align}
	Let $\gamma' = \gamma\mu$ and take $\mu\leq \frac{1}{2-\gamma}$, it is straightforward to show that $\frac{2}{1-\mu}\leq \frac{4}{1-\gamma'}$, and therefore,
	\begin{align}
	\sum_{\vec b\in \Lambda} \frac{1}{\norm{\vec a-\vec b}^\alpha}\frac{\norm{\vec b}^\beta}{e^{\gamma\norm{\vec b}}}
	\leq \frac{\lambda \left(\frac{4}{1-\gamma'}\right)^\alpha}{\ell^\alpha}
	+\frac{\lambda'\ell^{\beta+D-1}}{e^{\gamma' \ell}}.
	\end{align}
	Note that if we choose $\mu = \frac{1}{2-\gamma}$, then $\gamma'=\frac{\gamma}{2-\gamma}$ takes on a value between 0 and 1,  which can be arbitrarily close to 1.
\end{proof}
\subsection{Parameterizing a convex set}
\label{Subsec_para}
In this subsection, we show how we evaluate the sum over $\vec a$ in \cref{EQ_sumveca}. 
First, we parameterize a convex set by the distance to its boundary.
The following lemma simplifies a sum over every site in a convex set to a sum over the above distance, multiplied by the boundary area of the set.
\begin{lemma}
	\label{LEM_PARA}
	Let $A\subset \mathbb R^D$ be a compact and convex set in $\mathbb R^D$ with non-empty interior. Let $C\subset \mathbb R^D$ be another subset disjoint from $A$,
	and let $\ell = \dist{A,C}$ be the smallest distance between elements of the two sets.
	Furthermore, we denote by $\ell_{\vec a} = \dist{\vec a,C}$ the minimal distance from a given lattice site $\vec a$ in $A$ to $C$.
	For a decreasing function $f:\mathbb R\rightarrow \mathbb R$, we shall have
	\begin{align}
	\sum_{\vec a\in  A \cap \Lambda}f(\ell_{\vec a})\leq 2\eta \Phi( A)\sum_{\mu=0}^\infty f(\ell+\mu), 
	\end{align} 
	where $\eta$ is a constant that may depend only on $D$ and $\Phi( A)$ is the boundary area of $ A$.
\end{lemma}
\begin{proof}
	Let us divide the set $ A\in \mathbb R^D$ into disjoint subsets $ S_\mu = \left\{\vec a \in A:\mu \leq \dist{\vec a,\partial  A} \leq \mu+1 \right\}$ for $\mu=0,1,\dots$
	Note that the assumption that the interior of $A$ is non-empty implies that $\dist{\vec a,\partial A}$ is not uniformly zero. 
	Roughly speaking, $ S_\mu$ contains the sites in $ A$ whose distances to the boundary $\partial  A$ are between $\mu$ and $\mu+1$. 
	Therefore, $\ell_{\vec a}\geq \ell+\mu$ for all $\vec a\in  S_\mu$.
	We then obtain
	\begin{align}
	\sum_{\vec a\in  A \cap \Lambda}f(\ell_{\vec a}) &= \sum_{\mu=0}^{\infty} \sum_{\vec a\in  S_\mu\cap \Lambda}f(\ell_{\vec a})\\
	&\leq  \sum_{\mu=0}^\infty f(\ell+\mu)\Abs{ S_\mu\cap \Lambda},\label{EQ_sum_over_mu}
	\end{align}
	where $\abs{ S_\mu\cap \Lambda}$ is the number of lattice sites that lie within $ S_\mu$.
	
	Let $ A_\mu = \left\{\vec a \in A:\dist{\vec a,\partial  A} \geq \mu \right\}$ be a subset of $ A$ containing sites at least a distance $\mu$ from the boundary of $ A$. 
	Clearly, $ S_\mu =  (A_\mu\setminus  A_{\mu+1})\cup\partial  A_{\mu+1}$ and $\partial  S_\mu = \partial  A_\mu \cup \partial  A_{\mu+1}$.
	Roughly speaking, $ S_\mu$ is a shell with the outer surface $ A_\mu$, the inner surface $ A_{\mu+1}$ and a unity thickness.
	The number of lattice sites in $ S_\mu$ will be bounded by $\eta\Phi(S_\mu)=\eta (\Phi( A_\mu)+\Phi( A_{\mu+1}))$ (see \Cref{Subsec_count} for the definition of the constant $\eta$).
	Since $A$ is compact and convex, $\Phi( A_{\mu+1})<\Phi( A_{\mu})<\Phi( A)$ (see \Cref{Subsec_convex_shrinkable}). 
	Therefore, we arrive at the lemma.
\end{proof}
\subsubsection{The number of lattice sites in a compact region}
\label{Subsec_count}
In this subsection, we shall provide an upper bound on the number of lattice sites inside a compact set $A\subset \mathbb R^D$.
We use this bound in \cref{EQ_sum_over_mu} to estimate the number of lattice sites in the set $\Abs{ S_\mu\cap \Lambda}$ by its boundary area.
Let $A_{>}=\left\{\vec a\in A\cap \Lambda:\dist{a,\partial A}>\frac{1}{3} \right\}$ be the set of lattice sites that are at least a distance $\frac{1}{3}$ away from the boundary $\partial A$, and let $A_{\leq}=A\setminus A_{>}$ be the other lattice sites of $A$. 

First, note that for every lattice site $\vec a$ in $A_{>}$, there exists a $D$-ball $\B_{1/4}(\vec a)$ of radius $\frac{1}{4}$ that contains no other lattice site and $\B_{1/4}(\vec a)\subset A$.
Therefore, the number of lattice sites in $A_{>}$ is at most $\mathcal V(A)/\mathcal V(\B_{1/4}(\vec a)) =\eta_1 \mathcal V(A)$, where $\mathcal V(A)$ is the volume of $A$ in $\mathbb R^D$ and $\eta_1 = \mathcal V(\B_{1/4}(\vec a))^{-1}$.

Next, to count the lattice sites in $A_{\leq}$, we note that for every $\vec a\in A_{\leq}$, we can select a point $f(\vec a)\in\partial A$ on the boundary such that $\norm{f(\vec a)-\vec a}\leq \frac{1}{3}$.
We now argue that $\norm{f(\vec a)-f(\vec b)}\geq \frac{1}{3}$ for all distinct lattice sites $\vec a\neq \vec b$ in $A_{\leq}$.
Indeed, since $\vec a,\vec b$ are distinct lattice sites, the least distance between them is 1, i.e.\ $\norm{\vec a-\vec b}\geq 1$. Using a triangle inequality, we can show that
\begin{align}
	\norm{f(\vec a)-f(\vec b)}&\geq \norm{\vec a-\vec b}-\norm{f(\vec a)-\vec a}-\norm{f(\vec b)-\vec b}\nonumber\\
	&\geq 1-\frac{1}{3}-\frac{1}{3}=\frac{1}{3}.
\end{align}
Therefore, a $D$-ball $\B_{1/6}(f(\vec a))$ around $f(\vec a)\in \partial A$ shall contain no $f(\vec b)$ of any other lattice site $\vec b\in A_{\leq}$.
Therefore, the number of lattice sites in $A_{\leq}$ is at most $\eta_2\Phi(A)$, where $\Phi(A)=\abs{\partial A}$ is the boundary area of $A$ and $\eta_2$ is the area of a $(D-1)$-dimensional disk of radius $1/6$.  

In summary, the number of lattice sites in $A$ is therefore at most $\eta_1 V(A)+\eta_2\Phi(A)$.
In particular, for a shell $A$ whose volume $\mathcal V(A)$ can be upper bounded by $\eta_3\Phi(A)$, the number of lattice sites will be at most $\eta\Phi(A)$, where $\eta=\eta_1\eta_3+e\eta_2$.
\subsubsection{Convex sets in $\mathbb{R}^{D}$ are shrinkable}
\label{Subsec_convex_shrinkable}

In the proof of Lemma \ref{LEM_PARA} [see the discussion after \cref{EQ_sum_over_mu}], we used the fact that $\Phi( A_{\mu})<\Phi( A)$. In this section, we will show that this property of $A$\textemdash which we term \emph{shrinkability}\textemdash holds if $A$ belongs to the class of convex and compact sets in $\mathbb R^D$. The formal definition is as follows:
\begin{definition}[Shrinkable set]
	A compact set $A\subset \mathbb R^D$ with boundary $\partial A$ is \emph{shrinkable} if, for all $r>0$, $A_r = \left\{\vec a \in A:\dist{\vec a,\partial  A} \geq r \right\}$, we have that $\Phi(A_r)=\abs{\partial A_r}\leq \abs{\partial A}=\Phi(A)$.
\end{definition}
In other words, a set is shrinkable if the surface area of the boundary of $A_r\subseteq A$ is no larger than that $A$. 
In this section, we will prove that convexity is a sufficient condition for shrinkability. 
Recall that a set is compact if it is both closed and bounded, whereas convexity is usually defined as follows:
\begin{definition}
	A set $A$ is convex if for any $x,y\in A$ and any $\theta$ such
	that $0\le\theta\le1$, we have $\theta x+(1-\theta)y\in A$. 
\end{definition}
Examples of convex sets include $D$-balls and hyperrectangles, which are also shrinkable.
To prove this holds in general, we will first show that if $A$ is convex, then $A_r$ is also convex (or empty) for all $r>0$. 
To do this, we formulate an equivalent definition of a convex set as an intersection of halfspaces. 
\begin{definition}
	A halfspace $\mathcal{H}$ is given by the points $\{x\in\mathbb{R}^{D}\mid a^{T}x\ge b\}$,
	where $a\in\mathbb{R}^{D}\backslash\{0\}$. 
\end{definition}
From this definition, it follows that halfspaces are convex sets. 
A folk lemma \cite{Boyd04} states that a closed set $A$ is convex iff
\[
A=\bigcap_{k\in I}\{\mathcal{H}_{k}\mid\mathcal{H}_{k}\text{ halfspace},A\subseteq\mathcal{H}_{k}\},
\]
for some countable index set $I$. 
In other words, $A$ is equivalent to the intersection of all halfspaces that contain it. 
Since convexity is preserved under arbitrary intersection, this implies that $A$ is convex.
The converse follows from the separating hyperplane theorem\textemdash see \cite{Boyd04} for details.

With this equivalent definition of convexity in hand, we will prove that $A_r$ is also convex.
\begin{lemma}
	\label{AR_CONV}
	If a compact set $A\subset\mathbb{R}^{D}$ is convex, then $A_r = \left\{\vec a \in A:\dist{\vec a,\partial  A} \geq r \right\}$ is convex (or empty) for all $r>0$.
\end{lemma}
\begin{proof}
	Write $A$ as the intersection of half-spaces $\mathcal{H}_{k}=\{x\in\mathbb{R}^{D}\mid a_{k}^{T}x\ge b_{k}\}$,
	for $k\in I$. Then $A_{r}$ is the intersection of the half-spaces
	given by $H_{k}^{r}=\{x\in\mathbb{R}^{D}\mid a_{k}^{T}x\ge b_{k}+r\}$.
	By the converse of the above lemma, $A_{r}$ is convex (or empty).
\end{proof}
To show that $A$ is shrinkable, we must show that $\Phi(A_r) = |\partial A_{r}|\le|\partial A| = \Phi(A)$.
Following a standard technique in the literature, we define the \textit{nearest-point projection} of $\mathbb{R}^{D}$ onto a convex set and then show that it is a contraction. 
The following lemma implies that such a mapping is well-defined.
\begin{lemma}
	\label{CONV_PROJ}
	Given a non-empty, compact and convex set $A\subseteq\mathbb{R}^{D}$ and
	a point $x\in \mathbb{R}^{D}$, there exists a unique point $p_{A}(x)\in A$
	such that 
	\begin{align*}
	p_{A}(x) & =\underset{y\in A}{\arg \min}\|x-y\|.
	\end{align*}
\end{lemma}
\begin{proof}
	Since $A$ is compact, the continuous function $d_{x}(y)=\|x-y\|$ must achieve
	its minimum value on $A$. 
	
	Now suppose that minimum value of $d_x$ occurs at a point $y\in A$. We will
	show that $y$ is unique. Assume for the sake of contradiction that
	there exists some point $\tilde{y}\in A$ such that $d_{x}(y)=d_{x}(\tilde{y})$,
	but $y\neq\tilde{y}$. Then the set of points $x,y,$ and $\tilde{y}$
	form an isosceles triangle, with $\overline{y\tilde{y}}$ as the base.
	Dropping an altitude from $x$ intersects this line segment at the
	midpoint $m$ such that $\|x-m\|<\|x-y\|=\|x-\tilde{y}\|$. But $m$
	is a convex combination of $y$ and $\tilde{y}$, i.e.\ $m=\frac{1}{2}(y+\tilde{y})\in A$,
	so we have reached a contradiction. Thus, $y$ must be unique, and, therefore,
	$p_{A}(x)$ is well-defined. 
\end{proof}
The projection function $p_{A}(x)$ can be interpreted as generalizing
the concept of the orthogonal projection into an affine subspace. 
It is also well-known that the nearest point projection $p_A$ is a contraction mapping.
\begin{lemma}
	\label{CONTRACT_MAP}
	Given a nearest-point projection $p_{A}:\mathbb{R}^{D}\rightarrow A$ onto a convex set $A$, it holds 
	for all $x,y\in\mathbb{R}^{D}$ that 
	\begin{align*}
	\|p_{A}(x)-p_{A}(y)\| & \le\|x-y\|.
	\end{align*}
\end{lemma}
\begin{proof} 	
	While the lemma can be proved for all $x,y\in\mathbb{R}^{D}$, for
	our purposes, we only need to consider $x,y\notin A$. Assume that
	$p_{A}(x)\neq p_{A}(y)$. Then consider the hyperplanes $H_{x}$ and
	$H_{y}$ that pass through $p_{A}(x)$ and $p_{A}(y)$ respectively,
	and are perpendicular to the line segment $\overline{p_{A}(x)p_{A}(y)}$.
	(See the geometric diagram in Fig.~\ref{CON_PROJ}.)
	
	We prove by contradiction that $x$ ($y$) and $p_{A}(y)$ ($p_{A}(x)$) lie on opposite sides of $H_{x}$ ($H_{y}$). 
	Suppose without loss of generality that $x$ and $p_{A}(y)$
	lie on the same side of $H_{x}$. Then the point where the altitude from $x$
	intersects the line segment $\overline{p_{A}(x)p_{A}(y)}$ would lie in $A$,
	contradicting the fact that $p_{A}(x)$ is the nearest-point in
	$A$ to $x$. Thus, $x$ ($y$) must lie on the opposite side of $H_{x}$
	($H_{y}$) from $p_{A}(y)$ ($p_{A}(x)$). Then, as shown in Fig.~\ref{CON_PROJ},
	the points $x$ and $y$ must fall outside the rectangular strip between
	the two hyperplanes. From this we conclude that $\|p_{A}(x)-p_{A}(y)\|\le\|x-y\|.$
\end{proof}
\begin{figure}[t]
	\centering
	\includegraphics[width=0.35\textwidth]{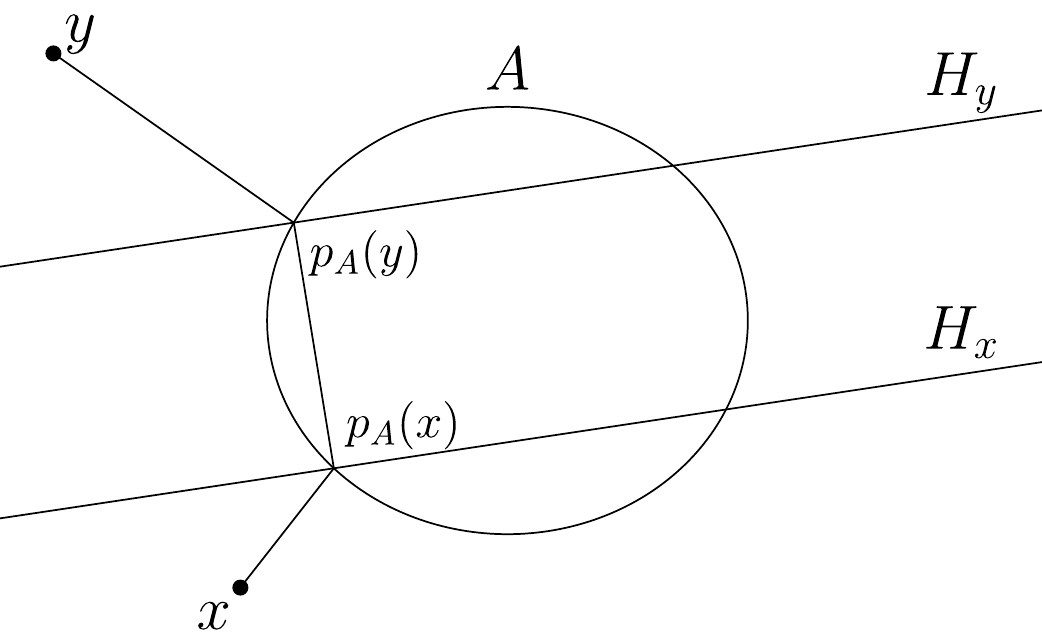}
	\caption{The nearest-point projection $p_{A}$ of two points $x$ and $y$ onto a
		compact set $A$ (oval). Also depicted are the line segment connecting
		the two image points $p_{A}(x)$ and $p_{A}(y)$, as well as the two
		hyperplanes orthogonal to it.}
	\label{CON_PROJ}
\end{figure}

The above result proves that the projection $p_{A}(x)$ is indeed
a contraction. Since contraction mappings do not increase lengths,
we can use this fact to demonstrate that the boundary of $A_{r}$
is less than that of $A$. 
\begin{theorem} If the set $A\subset\mathbb{R}^{D}$ is compact and convex, then $\Phi(A_r) = |\partial A_{r}|\le|\partial A| = \Phi(A)$. 
\end{theorem}
\begin{proof}
	Consider the projection $p_{A_{r}}:A\rightarrow A_{r}$. Note
	that for $r>0$, we have that $A_{r}=\{x\in A\mid d(x,A^{c})\ge r\}$
	is entirely contained in the interior of $A$, which implies that
	$A_{r}\cap\partial A=\emptyset$. Thus, our situation satisfies the
	assumption we made in the proof of Lemma \ref{CONTRACT_MAP}. 
	
	Under the action
	of $p_{A_{r}}$, any point in $\mathbb{R}^{D}$ outside of $A_{r}$ will
	get mapped to $\partial A_{r}$. In particular, since the map is onto, $\partial A$ will
	get mapped to $\partial A_{r}$, i.e.\ $p(\partial A)= \partial A_{r}$. 
	Using the fact that $p_{A_{r}}$ is contractive, we have that 
	\begin{align*}
	\Phi(A_r) =|\partial A_{r}| = |p(\partial A)| & \le|\partial A| = \Phi(A),
	\end{align*}
	from which we conclude that $A$ is a shrinkable set. 
\end{proof}
This provides the final step in our proof of Lemma \ref{LEM_PARA}.
Note that we do not require an explicit formula for the surface area
of the boundary of a $D$-dimensional convex set. 
In general, one may use the Cauchy-Crofton formula to calculate this quantity\textemdash for more details, see
Theorem 5.5.2 of Ref.~\cite{Klain97}. 

\bibliography{qu-simulation}

\begin{thebibliography}{57}
\bibitem [{Bos()}]{Boson-sampling-TBP}%
  \BibitemOpen
  \href@noop {} {\bibinfo  {journal} {To be published}\ }\BibitemShut {NoStop}%
\end{thebibliography}

\end{document}